\documentclass[twocolumn]{autart}

\usepackage{mathrsfs}
\usepackage{color}
\usepackage{amsfonts,amssymb}
\usepackage{subfigure}
\usepackage{booktabs}
\usepackage{amsmath}
\usepackage{enumerate}
\usepackage{algorithm}
\let\classAND\AND
\let\AND\relax
\usepackage{algorithmic}

\let\AND\classAND
\AtBeginEnvironment{algorithmic}{\let\AND\algoAND}

\usepackage{bm}
\usepackage{makecell}

\let\theoremstyle\relax
\usepackage{mathtools}
\usepackage{bbm}
\usepackage{dsfont}
\usepackage{pifont}
\usepackage{amsthm}
\usepackage{cite}
\usepackage{float}
\usepackage{enumitem}
\usepackage{balance}

\makeatletter
\long\def\@maketablecaption#1#2{\@tablecaptionsize
    \global \@minipagefalse
    \hbox to \hsize{\parbox[t]{\hsize}{\centering #1 \\ #2}}}

\makeatother

\newtheorem{definition}{Definition}
\newtheorem{theorem}{Theorem}

\newtheorem{remark}{Remark}
\newtheorem{proposition}{Proposition}
\newtheorem{assumption}{Assumption}

\DeclareMathOperator{\rank}{rank}
\DeclareMathOperator{\tr}{tr}

\begin{document}

\begin{frontmatter}

\title{Optimal Unpredictable Control for Linear Systems \thanksref{footnoteinfo}} 

\thanks[footnoteinfo]{
Preliminary results have been presented at the 2020 American Control Conference (ACC) \cite{li2020unpredictable}.}

\author{Chendi Qu}\ead{qucd21@sjtu.edu.cn},
\author{Jianping He}\ead{jphe@sjtu.edu.cn},    
\author{Jialun Li}\ead{jialunli@sjtu.edu.cn},   
\author{and Xiaoming Duan}\ead{xduan@sjtu.edu.cn} 
\address{Department of Automation, Shanghai Jiao Tong University, China}  


\begin{abstract}                          
In this paper, we investigate how to achieve the unpredictability against malicious inferences for linear systems. The key idea is to add stochastic control inputs, named as \emph{unpredictable control}, to make the outputs irregular. The future outputs thus become unpredictable and the performance of inferences is degraded. {The major challenges lie in: i) how to formulate optimization problems to obtain an optimal distribution of stochastic input, under unknown prediction accuracy of the adversary;} and ii) how to achieve the trade-off between the unpredictability and control performance.
We first utilize both variance and confidence probability of prediction error to quantify unpredictability, then formulate two two-stage stochastic optimization problems, respectively. Under variance metric, the analytic optimal distribution of control input is provided. With probability metric, it is a non-convex optimization problem, thus we present a novel numerical method and convert the problem into a solvable linear optimization problem. Last, we quantify the control performance under unpredictable control, and accordingly design the unpredictable LQR and cooperative control. {Simulations demonstrate the unpredictability of our control algorithm. The obtained optimal distribution outperforms Gaussian and Laplace distributions commonly used in differential privacy under proposed metrics.}
\end{abstract}

\end{frontmatter}

\setlength{\abovedisplayskip}{2pt} 
\setlength{\belowdisplayskip}{2pt}
\setlength{\parskip}{4pt}

\section{Introduction}
With the development of computation and communication, Cyber-Physical Systems (CPSs), such as mobile robots and smart grids, are promising to improve our life. However, these systems are prone to suffer from data leakage due to cyber and physical accessibility. When malicious agents obtain states or dynamics of these systems, they are able to infer the private data and design attacks \cite{han2018privacy, le2013differentially, mo2016privacy, he2018preserving}. As a result, data privacy and security become prominent issues in CPSs.

This paper aims to achieve an unpredictability of future outputs (or concerned states) of a CPS with linear dynamics, i.e., making an adversary hard to predict outputs of the system precisely. Inspired by recent works on protecting privacy by adding noise \cite{han2018privacy}, we add a stochastic term in the control input, named as an unpredictable control, to maximize the unpredictability. We propose variance and probability metrics to quantify unpredictability, aiming to find the optimal noises design for the systems.

\subsection{Motivations and Challenges}
{The disclosure of states is sensitive to a CPS. If adversaries predict the future outputs accurately, they are able to design severe attacks. For example, consider a mobile robot moving outdoor to monitor the surrounding environment. If the future position of the robot is known by an attacker, the attacker is able to intercept the robot precisely \cite{manyam2019optimal, qu2022moving, li2019learning}. Therefore, preserving privacy of states is critical. Since a CPS evolves with its dynamic equations, attackers are able to predict the future state with past observed trajectory. A common approach to protect the states from malicious inference is adding random disturbances \cite{le2013differentially,han2018privacy}, which increases the variance of inference error and decreases the probability of adversary making accurate prediction.} The first question we need to consider is \emph{how to quantify the unpredictability}? To address this, we propose two metrics, including variance of the prediction error (expectation of the square of the prediction error) and the probability of precise prediction given a certain range. 
{Besides, we are interested in \emph{what is the optimal input distribution to protect unpredictability} in the sense of these metrics?} 

The state unpredictability problem is similar to traditional anti-predator behaviors in biology and classical pursuit-evasion games, but novel and more challenging. First, researches on anti-predator behaviors emphasize on explaining the escape mechanism according to defined metrics and mechanistic \cite{roberts2004positional}. In \cite{herbert2017escape}, an embedding matrix of history trajectories over a time window is established and the entropy of singular values of the embedding matrix is taken to evaluate path complexity. This 
metric is formed with history data and thus hard to be optimized directly. Therefore, it is difficult to be used to design the optimal future anti-predator behaviors. Second, in pursuit-evasion games, the interactions between pursuers and evaders are modeled as differential equations. Then, an optimal control problem is set up based on the deterministic model \cite{isaacs1999differential}\cite{mejia2019solutions}, which is simpler than our problem.

The challenges of the concerned problem are twofold. On the one hand, the measurement accuracy and prediction algorithms of potential attackers are unknown. The designed control method is expected to protect the future state against various prediction methods. On the other hand, with probability metric, it is a two-stage non-convex optimization problem. This problem is hard to be solved by existing solvers. 

\subsection{Key Idea and Contributions}

Inspired by recent works on adding random noises to preserve the secrecy of CPSs, we propose a method by adding an extra random input in original input to make the system states unpredictable. We use the variance of prediction errors and confidence probability to quantify unpredictability. The optimal distributions of the extra input are expected to maximize the unpredictability. 

To make our method applicable to different types of prediction algorithms by attackers, we formulated a max-min (min-max) optimization problem to optimize the worst case that attackers make accurate predictions. When utilizing the confidence probability metric as the objective function, the general problem is non-convex as an Optimal Uncertainty Quantification (OUQ) problem. To solve this problem, we first discretize the original problem in a high-dimensional spherical coordinate system and then convert the problem into a convex linear optimization problem, which can be solved by existing solvers easily. 

{The differences between this paper and our conference version \cite{li2020unpredictable} include i) definitions, problem formulations and main results of a single-integral model have been extended to general linear systems, ii) a novel numerical algorithm to achieve an optimal solution to the problem under probability metric is provided, iii) the newly arisen problem that how to calculate the extra unpredictable control input is solved, iv) the proposed secure control method is combined with LQR and cooperative control methods, v) extended simulations of the numerical algorithm are provided.} The main contributions are summarized as follows. 

\begin{itemize}
\item We propose variance and confidence probability metrics of the prediction error to quantify the unpredictability. Then we formulate the unpredictability preserving problem as two-stage optimization problems. The obtained unpredictable control generalizes well to various prediction algorithms of adversaries.

\item The analytic input distribution solution under variance metric is obtained. With probability metric, we present a novel numerical method and convert the problem into a solvable linear optimization problem. {The solved optimal distribution outperforms commonly used Gaussian and Laplace distributions under proposed metrics.} 

\item We quantify the control performance degradation caused by stochastic inputs. Taking LQR control and cooperative control as examples, we demonstrate how to achieve the optimal trade-off between output unpredictability and control performance.

\item {It is revealed that probability metric is better than variance metric to determine an optimal distribution. By optimizing the probability metric we obtain a specific form of the optimal distribution, while the variance metric only gives the relationship between unpredictability and covariance of input. Nevertheless, the covariance serves as a bridge to achieve a trade-off between security and control performance. This provides inspirations on which metric should be chosen in other privacy-preserving problems of CPS.}
\end{itemize}

\subsection{Related Works}
\subsubsection{Security and privacy of CPSs from the control perspective}
The security of CPSs includes three main attributes as confidentiality, integrity and availability (CIA). From the control-oriented perspective, researches in this area are divided into two aspects. One studies impacts of different attacks on control performance and aims to design detection and control mechanisms to guarantee the performance of the system \cite{teixeira2010cyber, liu2011false, pasqualetti2013attack, mo2009secure, irita2017detection, teixeira2015secure, zhao2019resilient}. This aspect is related to integrity and availability. Another investigates on disclosure of sensitive data in system and protection of sensitive information from unauthorized users, which aims to guarantee the confidentiality attribute. Related works quantify how accurate the attacker can infer the concerned states of the system. \cite{he2018preserving} presents a novel metric to analyze data privacy and studies privacy disclosure of initial states in consensus process by adding noise.{\cite{satchidanandan2016dynamic} proposes a ``watermarking" technique and injects it into the CPS as an excitation to reveal the malicious attack.} \cite{le2013differentially} studies protection of the true state trajectories from estimations with output data and proposes a differential-privacy-preserving filter by adding noise in published data. \cite{mo2016privacy} studies privacy of the initial states of the multi-agent system in consensus. The privacy is achieved by adding noise in an iterative process. 
These works protect confidentiality by adding noise to increase uncertainty of available data to make adversaries hard to infer private data accurately. 
{Differential privacy can also reveal how the state (private variable) is protected. \cite{geng2015optimal} considers a fundamental trade-off between data privacy and utility. \cite{duncan2000optimal} focuses on protecting the important data from tracker attacks through additive noises and derives an optimal disclosure limitation strategy.} 
There are also some privacy-preserving methods based on information-theoretic approaches \cite{nekouei2019information}.

Different from existing works, we study on quantifying privacy of future outputs (or concerned states) of linear dynamics under designed metrics, defined as output (states) unpredictability, instead of the initial or past sensitive states in the differential privacy studies. The motivation is to protect future data from malicious predictions by adversaries. We take a similar method by adding random control component in original control input, but give an optimal distribution of the unpredictable control component by solving a newly formulated  problem.

\subsubsection{Prediction of linear systems}
Predicting the states of dynamic system is a classic problem in various fields. In finance, researchers predict the market with spectrum and use various models like auto-regressive integrated moving average (ARIMA) model \cite{ariyo2014stock} to describe the phenomenon. In climatology, the weather is described by a series of complicated nonlinear systems. The future states of these systems are hard to be predicted accurately because the prediction errors increase sharply with initial observation errors \cite{delsole2004predictability,delsole2005predictability}. On the other hand, in traffic and autonomous driving, the future trajectories of traffic participants are expected to be predicted accurately. The neural networks, e.g., LSTM \cite{alahi2016social} and Graph Convolutional Networks \cite{shi2021sgcn,li2021spatio}, are promising to solve these problems. 
As for predicting general linear systems, the most famous method is Kalman filter. Kalman filter gives the optimal estimation or prediction of the system with minimum variance of estimation error. When the noise in the system is Gaussian, it is also optimal in probability metric, i.e. the state with maximum probability. However, when the distribution of noise is non-Gaussian, the classical Kalman filter may not be optimal with probability metric\cite{hanzon2001state}. 

In the problem we are concerned, to make the future states of the system unpredictable, the stochastic input can be designed and the optimal form is to be optimized, which may not be a Gaussian distribution. Hence, we first analyze the prediction error with both variance and probability metric and achieve optimal distribution in these two cases respectively. 

\subsection{Organization}
This paper is organized as follows. In Section \ref{secii}, some preliminaries are provided and the unpredictability problem is formulated as random optimization problems. In Section \ref{seciii}, optimization problems are solved and the optimal control is designed. Section \ref{seciv} illustrates how to combine our method with the existing control method. Section \ref{secv} shows the simulation results and Section \ref{secvi} concludes. 

\section{Preliminaries And Problem Formulation}\label{secii}

\subsection{Nominal Control and Prediction Models}\label{ii-a}
A discrete linear time-invariant system is described by  
\begin{equation}\label{eq1}
\begin{aligned}
x(k+1) &= Ax(k)+Bu(k),\\
y(k) &= Cx(k), k \in \mathbb{N}, 
\end{aligned}
\end{equation}
where $x(k) \in \mathbb{R}^{n}, u(k) \in \mathbb{R}^{p}$ and $y(k) \in \mathbb{R}^{m}$ denote the state, input and output, respectively.  
Then, we have 
\begin{equation*}
y(k+\tau) \!=\! Cx(k+\tau) \!=\! CA^{\tau}x(k) + \sum_{l=1}^{\tau}CA^{l-1}Bu(k+\tau-l). 
\end{equation*}

{\bf{Prediction-based Attack Model}}: {Suppose that there is an attacker, aiming to predict future outputs of system \eqref{eq1} by historical measurements of the outputs. 
Assume the attacker has knowledge of the system model with parameter matrices $A, B, C$. The attacker has access to the output $y(k)$ with potential measurement noises, and the system input $u(k)$ is unknown. }

Let $\varepsilon(k)$ be the error of the posterior estimate of $x(k)$, i.e., $\varepsilon(k)=x(k)-\hat x(k)$, which is relevant to prediction accuracy and subsequent unpredictable optimal control design. Suppose that the attacker can get an unbiased state estimate based on the output and system model, i.e., $\mathbb{E}\{\varepsilon(k)\}=0$. To facilitate the discussion, we divide the attacker's estimation into two situations.
\begin{itemize}
\item {\textbf{Condition 1} ($\bm{\mathrm{C}_1}$, $\varepsilon(k)\equiv0$): 
The state $x(k)$ is accurately estimated by the attacker, that is the posterior estimate $\hat x(k)=x(k)$.}
\item \textbf{Condition 2} ($\bm{\mathrm{C}_2}$, $\varepsilon(k)$ is a random vector): The expectation of $\varepsilon(k)$ equals to zero. This means that
the expectation of $\hat x(k)$ equals to $x(k)$. {In this condition,
$\varepsilon (k)$ and $u(k+\tau-1)$ are independent with each other for each $\tau \in \mathbb{N^{+}}$.
{Both the probability density function (PDF) and the variance of $\varepsilon$ are unknown}, where the variance is denoted by
$\mathbb{D}(\varepsilon)$. }
\end{itemize}

\begin{remark}
{Note that considering a precise estimation actually makes our unpredictability preserving more difficult, since the estimation variance will be the lowest when $\varepsilon(k)=0$ (see this from proof in Appendix \ref{app0}). For the plausibility of the above assumptions, we provide a possible case: When matrix $CB$ has full column rank and the attacker observes an accurate initial state $x_0$ (which is feasible for mobile agents since they usually keep stationary at the initial state), the accurate state estimation can be obtained.}
\end{remark}

{Suppose the attacker follows the linear prediction step based on known model dynamics as
\begin{equation}
\hat{y}(k+\tau|k) = CA^{\tau}\hat{x}(k|k) + \sum_{i=1}^{\tau}CA^{i-1}B\hat{u}(k+\tau-i|k),
\end{equation}
{where $\hat{y}(k+\tau|k)$ and $\hat{u}(k+\tau-i|k)$ are prior predictions of $y(k+\tau), {u}(k+\tau-i)$, respectively}, and $\hat{x}(k|k)$ is the posterior estimation. $\hat{u}(k+\tau-i|k)$ and $\hat{x}(k|k)$ are obtained by attacker according to the measurements before $k$-th time instant and we do not specify the exact estimation or prediction methods.} We omit the symbol for the current time instant $k$, e.g., denoting $\hat{u}(\cdot) = \hat{u}(\cdot|k)$ and $\hat{x}(k) = \hat{x}(k|k)$ for simplicity.

The prediction error of attacker is described by
\begin{equation}\label{eq10}
\varepsilon_{y}(k+\tau|k)=y(k+\tau)-\hat y(k+\tau|k).
\end{equation}
In this work, we take the case of $\tau=1$ as basis, i.e., consider the one step prediction error $\varepsilon_{y}(k+1|k)$.

\subsection{Unpredictable Control Model}
Our goal is to design control input such that the system output cannot be predicted accurately while maintaining control performance. To increase the prediction error, an extra control input $u_{e}(k)$ is added to the input $u(k)$. {Thus, 
the system state equation is changed to   
\begin{equation}\label{eq:system_dynamic}
\begin{aligned}
x(k+1) &= Ax(k)+B(u(k)+u_e(k)).
\end{aligned}
\end{equation}
Let $\varepsilon_y(k+1)= \varepsilon_y(k+1|k)$. Then, the one-step prediction error satisfies
\begin{equation}\label{predict_error}
\begin{aligned}
\varepsilon_y(k+1) &= Cx(k+1)- C(A\hat{x}(k)+B \hat{u}(k))\\&=  CA\varepsilon(k)+CB(u(k) - \hat{u}(k))+CB u_e(k)\\
&= CA\varepsilon(k) + B_1(u(k) - \hat{u}(k)) + \theta(k),
\end{aligned}
\end{equation}
where $B_1 = CB$ and  $\theta(k) = C B u_e(k)\in \mathbb{R}^m$.
For simple statement, we let 
\begin{equation}\label{extra_input}
\theta_e(k) = Bu_{e}(k),
\end{equation}
then we have  $\theta (k)=C \theta_e(k)$.}

{If $u_{e}(k)$ is a function related to time, the outputs are a series of data about time.
In this situation, it is not difficult to predict or regress the outputs by methods like ARIMA \cite{liu2016online} or RNN \cite{zhang1998time}. However, if $u_{e}(k)$ is chosen as a random vector sequence satisfying certain distribution, then outputs are also random and difficult to be predicted accurately based on history data.} In addition, existing results show that random sequences are better than chaotic sequences \cite{hilborn2000chaos}. Therefore, the randomness of $u_{e}(k)$ is leveraged to make the outputs unpredictable. 

Let the PDF of $u_{e}(k)$ be $f_u(z)$, which satisfies 
\begin{itemize}
    \item {$f_u(z)$ is symmetrical about each component of $u_{e}(k)$, i.e. for $\forall\,i=1, \cdots, p$, we have 
\begin{equation}
f_u(z_1,\cdots , z_i, \cdots, z_p) = f_u(z_1,\cdots, -z_i, \cdots, z_p). 
\end{equation}
It follows immediately that $u_e(k)$ is zero-mean.}
\item The distributions of components of $u_{e}(k)$ are independent of each other. 
\item The variance of each $u_{e,i}(k)$ is bounded, i.e.,
\begin{equation}\label{extra_input_variances}
\mathbb{D}(u_{e,i}(k))\leq\sigma_{u,i}^2,
\,i= 1, \cdots, p.
\end{equation}
\end{itemize}
{For the above assumptions, considering those most commonly used noise distributions in privacy preserving are symmetric, such as the Gaussian and Laplace distributions \cite{nozari2017differentially,kawano2020design}, the symmetry assumption here is reasonable.
The independence of the components in $u_e(k)$ facilities the subsequent analysis and is reasonable since we usually add noise to each component of $u(k)$ individually considering the mobile agent control network systems.}
{Let $f_{\theta}(z) = f_{\theta}(z_1, \cdots, z_m)$ be the PDF of 
$\theta(k)$. {Then,
$f_{\theta}(z)$ is symmetrical about the vector $z$, i.e., $f_{\theta}(z)=f_{\theta}(-z)$, and
\begin{equation}
f_{\theta}(z_1,\cdots, z_i, \cdots z_m) = f_{\theta}(z_1,\cdots, -z_i, \cdots, z_m). 
\end{equation}}

According to the third assumption, the covariance matrix of $\theta(k)$ is element-wise bounded. Let  $\Sigma \geq 0$ be the covariance matrix of $\theta(k)$, where
$\Sigma_{ij}=\mathbb{E}\left(\theta_i\theta_j\right)$. Define the least upper bound matrix of $\Sigma$ as $\overline{\Sigma}$, i.e. we have $\Sigma_{ij}\leq\overline{\Sigma}_{ij},\forall\,i,j = 1,2,\cdots, m$. The {greatest} lower bound of $\Sigma$ is defined as $\underline{\Sigma}$. Clearly, we have
\begin{equation}\label{theta_PDF_condition2}
\underline{\Sigma} \leq \Sigma \leq \overline{\Sigma}.
\end{equation}
From \eqref{extra_input} and \eqref{extra_input_variances},  $\underline{\Sigma}$ and $\overline{\Sigma}$ can be chosen as
\begin{equation}\label{sigma_theta}
\begin{aligned}
\underline{\Sigma}_{ij} &= -\sum_{l=1}^{p} (-b_{il}b_{jl})^{+}\sigma_{u,l}^{2},\\ \overline{\Sigma}_{ij} &= \sum_{l=1}^{p}(b_{il}b_{jl})^{+}\sigma_{u,l}^{2},
\end{aligned}
\end{equation}
where $(c)^{+}=\max\{c, 0\}, c\in \mathbb{R}$.

Considering unpredictable control, the objective is changed to finding an optimal distribution, $f^{*}_{\theta}(z)$, of $\theta$ such that the prediction error is maximized.}

\subsection{Problem of Interest}
Note that the norms of $\varepsilon_{y}(k+1)$ cannot be optimized directly due to its randomness. Thus, we introduce the variance of the predict error $\mathbb{E}(\|\varepsilon_{y}\|_2^{2})$ and confidence probability metric $\mathbb{P}(\|\varepsilon_{y}\|_{2}^{2}\leq \alpha^{2})$ (the probability of a given accuracy prediction), respectively, to quantify the output unpredictability of a system. {Comparing to the traditional metric of differential privacy, the proposed variance and confidence probability metrics provide more exact unpredictability degree of the noise-adding mechanism in the face of state prediction, especially when the prediction accuracy is concerned.}

Then, we determine the optimal $f_{\theta}(z)$ by maximizing these two unpredictability metrics. $u(k)$ in \eqref{eq:system_dynamic} is assumed to be unknown to the attacker. We formulate the following two two-stage optimization problems. 
\begin{equation}\label{eq11}
\begin{aligned}
\bm{\mathrm{P_1}}:\qquad\max\limits_{f_\theta(z)}&\min\limits_{\hat u(k)}\mathbb{E}\,(\|\varepsilon_{y}(k+1)\|_2^{2})\\
\text{s.t.} \  &f_{\theta}(z) = f_{\theta}(-z),\, \underline{\Sigma}\leq \Sigma \leq \overline{\Sigma},\\
\end{aligned}
\end{equation}
\begin{equation}\label{eq11.2}
\begin{aligned}
\bm{\mathrm {P_2}}:\qquad\min\limits_{f_\theta(z)}&\max\limits_{\hat u(k)}\mathbb{P}\,(\|\varepsilon_{y}(k+1)\|_{2}^{2}\leq\alpha^{2})\\
\text{s.t.} \ &f_{\theta}(z) = f_{\theta}(-z),~\underline{\Sigma}\leq \Sigma \leq \overline{\Sigma}.\\
\end{aligned}
\end{equation}
where $\hat u(k)$ is the input prediction, $~\alpha\in\mathbb{R^{+}}$.
The first metric $\mathbb{E}(\|\varepsilon_{y}\|_2^{2})$ reflects the mean square of prediction error, i.e., the variance. The second metric $\mathbb{P}(\|\varepsilon_{y}\|_{2}^{2}\leq\alpha^{2})$ denotes the probability that the prediction accuracy satisfies the preset range, i.e., confidence probability. {The modeling of $\bm{\mathrm {P_1}}$ and $\bm{\mathrm {P_2}}$ can be viewed as optimizing the worst situations for the system. Inner problems optimize $\hat{u}(k)$ to evaluate the smallest $\mathbb{E}(\|\varepsilon_{y}\|_2^{2})$ 
and the largest $\mathbb{P}(\|\varepsilon_{y}\|_{2}^{2}\leq\alpha^{2})$ as the best prediction the adversary can achieve, then we make the prediction less reliable through the outer problems.}

The formulated problems $\bm{\mathrm {P_1}}$ and $\bm{\mathrm {P_2}}$ are hard to solve, since the objective functions are related to unknown posterior estimate error $\varepsilon(k)$. Therefore, we discuss the solutions of these two problems according to two situations of $\varepsilon(k)$ listed in Sec. \ref{ii-a}. Besides, $\bm{\mathrm {P_2}}$ is non-convex as an Optimal Uncertainty Quantification (OUQ) problem \cite{owhadi2013optimal, han2015convex}. Different from general OUQ problems, $\bm{\mathrm {P_2}}$ is two-stage and cannot be solved by the state-of-the-art algorithms of solving OUQ problems.

\section{Unpredictable Control Method}\label{seciii}
In this section, we give the optimal forms of $u_e$ and $\theta$ for $\bm{\mathrm {P_1}}$ and $\bm{\mathrm {P_2}}$ to make the outputs of system unpredictable. 
\subsection{Optimal Distribution of $\bm{\mathrm {P_1}}$}
Mathematically, we first give the definition of the optimality in terms of the attacker's prediction and distribution of $\theta$.
\begin{definition}
\emph{(Optimal input prediction)}
With variance metric, when $J_1=\mathbb{E}(\|\varepsilon_{y}(k+1)\|_2^{2})$, if $\forall
\hat u(k)\in \mathbb{R}^{p},$
$$J_1(f_\theta(z),\hat u(k),\hat x(k))\geq J_1(f_\theta(z),\hat u^{*}(k),\hat x(k)),$$
$\hat u^{*}(k)$ is an optimal input prediction with respect to $\hat x(k)$.
\end{definition}
\begin{definition}
\emph{(Optimal distribution)}
With variance metric, if arbitrary $f_\theta(z)$ satisfies
$$ J_1(f_\theta(z),\hat u^{*}(k),\hat x(k))\leq J_1(f^{*}_\theta(z),\hat u^{*}(k),\hat x(k)),$$ 
then $f^{*}_\theta(z)$ is an optimal distribution.
\end{definition}
We provide the following theorem as the solution to $\bm{\mathrm {P_1}}$.
\begin{theorem}\label{thm1}
{For both $\bm{\mathrm{C}_1}$ and $\bm{\mathrm{C}_2}$, $f_\theta(z)$ is an optimal distribution for $\bm{\mathrm {P_1}}$ iff
\begin{equation}
\Sigma_{ii} = \mathbb{D}(\theta_i)=\sigma_i^{2},\, i = 1,2, \cdots, m,
\end{equation}
where $\sigma_i^{2} = \overline\Sigma_{ii}$.}
\end{theorem}

\begin{proof}
The proof is given in Appendix \ref{app0}.
\end{proof}
\begin{remark}
Theorem \ref{thm1} indicates that the larger the variances are, the harder an attacker makes precise predictions. {This is consistent with our intuitions since larger variance means higher irregularity. }
\end{remark}

{$\mathbb{E}(\|\varepsilon_{y}(k+1)\|_2^{2})$ represents the mean deviation between actual and predicted positions. To minimize this index, the variances of $\theta_i$ are all expected to be largest, but the specific PDFs of $\theta_i$ can not be obtained.} These solutions to $\bm{\mathrm {P_1}}$ have different values according to other metrics, such as $\mathbb{D}(\|\varepsilon_{y}\|_2^{2})$.
The larger $\mathbb{D}(\|\varepsilon_{y}\|_2^{2})$ will lead to larger fluctuations, which means that attackers are able to make more accurate predictions
sometimes.

\subsection{Optimal Distribution of $\bm{\mathrm {P_2}}$}
Next, we leverage the probability measure and solve problem $\bm{\mathrm {P_2}}$.
Before giving an optimal solution to $\bm{\mathrm {P_2}}$, we provide some definitions as follows. 
\begin{definition}\label{def4}
\emph{(Optimal input prediction)}
With probability metric, let $J_2=\mathbb{P}(\|\varepsilon_{y}(k+1)\|_{2}^{2}\!\leq\!\alpha^2)$, 
if $\forall\,\hat u(k)\in \mathbb{R}^{p}$ satisfies
\begin{align*}
J_2(f_{\theta}(z),\hat u(k),\hat x(k),\alpha)\leq J_2(f_{\theta}(z),\hat u^{*}(k),\hat x(k),\alpha),
\end{align*}

\!\!then $\hat u^{*}(k)$ is an optimal input prediction in respect to $\hat x(k)$ in the sense of the confidence probability.
\end{definition}
\begin{definition}\label{def5}
\emph{(Optimal distribution)}
With probability metric, if an arbitrary PDF vector $f_{\theta}$ satisfies
\begin{align*}
J_2(f_{\theta}(z),\hat u^{*}(k),\hat x(k),\alpha)\geq J_2(f^{*}_{\theta}(z),\hat u^{*}(k),\hat x(k),\alpha),
\end{align*}
\!\!then $f^{*}_{\theta}(z)$ is an optimal distribution.
\end{definition}

With the probability measure as the objective function, the optimal distribution $f^{*}_{\theta}(z)$ can be obtained under condition $\bm{\mathrm{C}_1}$ $\hat{x}(k) = x(k)$. We have
\begin{equation}
\begin{aligned}
J_2=\,&\mathbb{P}(\left\|y(k+1)-\hat y(k+1)\right\|_{2}^{2}\leq\alpha^{2})\\
=\,&\mathbb{P}(\left\| \theta(k)+ B_1u(k) - B_1\hat u(k) \right\|_{2}^{2}\leq\alpha^{2})\\
=\,&\int_{\Omega} f_{\theta}(z)\ \mathrm{d}{z},
\end{aligned}
\end{equation}
{where $\Omega=\{\theta\in\mathbb{R}^m|\big\|
\theta-\widetilde u\| \leq \alpha\}$ and 
\[\widetilde u = [\widetilde u_1, \cdots, \widetilde u_m]^T = B_{1}\hat u(k) - B_{1}u(k).\]
Then $\bm{\mathrm{P_2}}$ can be rewritten as
\begin{equation}\label{probabilitz_problem_case1}
\begin{aligned}
\bm{\mathrm {P_2^0}}:\qquad\min\limits_{f_\theta(z)}&\max\limits_{\hat u(k)}\ \int_{\Omega} f_{\theta}(z)\ \mathrm{d}{z}\\
\text{s.t.} \ &f_{\theta}(z)=f_{\theta}(-z),~ \underline{\Sigma}\leq \Sigma \leq \overline{\Sigma}.
\end{aligned}
\end{equation}
Now optimizing $\hat{u}(k)$ is to find an optimal $\Omega$ for the inner problem.} The following theorem provides a solution.
\begin{theorem}\label{thm2}
Under $\bm{\mathrm{C}_1}$, if $f_{\theta}(z)$ is continuous, then there exist a  constant $\alpha\in (0, \sqrt{3}\min \limits_{i}{\sigma_i}]$,  such that $f^*_{\theta}(z)$ is the multivariate uniform distribution with finite maximum variances, 
 where
\begin{equation}\label{uniform_theta}
\!f^{*}_{\theta}(z)\!=\!\left\{
\begin{aligned}
\centering
\frac{1}{(2\sqrt{3})^m\prod \limits_{i = 1}^{m}\sigma_i}\,&~
\text{if}~z\in [-\sqrt{3}\sigma_i,\sqrt{3}\sigma_i], \\
0,~~~~&\text{otherwise},
\end{aligned}
\right.
\end{equation}
i.e., $f^*_{\theta}(z)$ is an optimal distribution with the confidence probability measure.
\end{theorem}
\begin{proof}
The proof is given in Appendix \ref{app1}.
\end{proof}

Theorem \ref{thm2} proves that there exists $\alpha$ in a small range such that the multivariate uniform distribution is optimal in the set of continuous functions. When $\alpha$ takes arbitrary values and $f_{\theta}(z)$ is not continuous, the uniform distribution is not optimal and $f^{*}_{\theta}(z)$ is related to the parameter $\alpha$. 

Next, a numerical method to solve $\bm{\mathrm {P^0_2}}$ under $\bm{\mathrm{C}_1}$ is provided. We first confine the range of $\theta$ as 
\[\Omega_1=\left\{\theta\Big| \|\theta\| \leq a, a > \alpha \right\}.\]
Since the input of a system is bounded, this operation is reasonable in practice. By setting boundary $a$ appropriately, the probability that $\theta$ locates outside $\Omega_1$ is smaller than arbitrary given positive parameter. This is guaranteed theoretically by Markov's inequality as
\begin{equation}\label{eq:sphere_boundary}
\mathbb{P}\left(\|\theta\|\geq a\right) = \mathbb{P}\left(\|\theta\|^2_2\geq a^2\right)\leq 
\frac{1}{a^2}\mathbb{E}(\theta^{T}\theta)=\frac{1}{a^2}\sum_{i=1}^{m}\sigma^2_i.
\end{equation}
If $a=5\sqrt{\sum_{i=1}^{m}\sigma^2_i}$, then $\mathbb{P}\left(\|\theta\|< a\right)\geq0.96$.

{\bf{Simplify to One-stage Optimization Problem}}:
For the convenient representation of integral domain, we discretize, $\Omega_1$, into pieces in the high dimensional sphere coordinate system instead of the Cartesian coordinate system. 

When $\widetilde u \neq 0_m$, it is still hard to represent $\Omega$ by the sum of divided pieces and do integral. Hence,  to simplify $\bm{\mathrm {P_2}}$, we give the following theorem to set $\widetilde u = 0$ reasonably.

\begin{theorem}\label{thm3}
Under $\bm{\mathrm{C}_1}$, for $\forall\, \alpha\in \mathbb{R}^{+}$, one optimal solution to $\bm{\mathrm {P_2}}$ is $(f^{*}_{\theta}(z), \hat u^{*}(k)) = (f^{*}_{\theta}(z), 0)$, where 
\[f^{*}_{\theta}(z) = \arg\min_{f_{\theta}(z)}\int_{\Omega_0} f_{\theta}(z)\ \mathrm{d}{z},\]
where $\Omega_0=\left\{\theta \big |\|\theta\| \leq \alpha\right\}\subset \Omega_1$.
\end{theorem}
\begin{proof}
The proof is provided in Appendix \ref{app2}.
\end{proof}

Theorem \ref{thm3} illustrates that among the optimal solutions of $\bm{\mathrm {P^0_2}}$, there exists an solution satisfying $\hat u^{*}(k) = 0$. A sufficient and necessary condition to make $\hat u^{*}(k) = 0$ hold is that
\begin{equation}\label{inner_opt_condition}
\int_{\Omega_0} f_{\theta}(z)\ \mathrm{d}{z} \geq \int_{\Omega} f_{\theta}(z)\ \mathrm{d}{z}, \forall\, \hat{u}(k)\in \mathbb{R}^p.
\end{equation}
Then, we are able to set $\hat u = 0$ and add \eqref{inner_opt_condition} to $\bm{\mathrm {P_2^0}}$. 
Hence, the two-stage optimization problem is simplified to a one-stage optimization problem as
\begin{equation}\label{P2-3}
\begin{aligned}
\bm{\mathrm {P_2^1}}:\qquad\min\limits_{f_\theta(z)}&\ \int_{\Omega_0} f_{\theta}(z)\ \mathrm{d}{z}\\
\text{s.t.} \ &f_{\theta}(z)=f_{\theta}(-z),\  \underline{\Sigma}\leq \Sigma \leq \overline{\Sigma}, (\ref{inner_opt_condition}).
\end{aligned}
\end{equation}

In $\bm{\mathrm {P_2^1}}$, the constraint \eqref{inner_opt_condition} integrals on $\Omega$ with $\hat u(k) \neq 0$. When the domain is discretized in the spherical coordinate system, the integral domain is hard to calculate as sum of units with $\hat u(k)$ taking arbitrary value except zero. The following proposition gives a sufficient condition to let \eqref{inner_opt_condition} holds. By the proposition, \eqref{inner_opt_condition} is replaced with a new constraint on decision variable $f_{\theta}(z)$. Although the new constraint makes the solution suboptimal, the converted optimization problem is solvable with the state-of-art solvers. 

\begin{proposition}\label{propo1}
If $\forall\,\|z_1\| \leq \|z_2\|$, $f_{\theta}(z_1)\geq f_{\theta}(z_2)$ holds true, 
then we have $\hat u^{*} = \arg \max_{\hat u(k)}{\int_{\Omega}f_{\theta}(z)}= 0$.
\end{proposition}

By proposition \ref{propo1}, the converted problem $\bm{\mathrm {P^{2}_2}}$ is
\begin{equation}\label{P2_approximation}
\begin{aligned}
\bm{\mathrm {P^{2}_2}}:\qquad\min\limits_{f_\theta(z)}\,
&\int_{\Omega_0} f_{\theta}(z)\ \mathrm{d}{z}\\
\text{s.t.} &\ f_{\theta}(z)=f_{\theta}(-z),\   \underline{\Sigma}\leq \Sigma \leq \overline{\Sigma},\\
&f_{\theta}(z_1) \geq f_{\theta}(z_2), ~\forall~ \|z_1\| \leq \|z_2\|.
\end{aligned}
\end{equation}

{\bf{Simplify to Solvable Linear Optimization Problem}}:
In order to discretize $\bm{\mathrm {P_2^2}}$ and integral on spherical region $\Omega$, we convert the coordinates $[\theta_1, \cdots, \theta_m]^{T}$ in Cartesian coordinate system into high-dimensional sphere coordinate by
\begin{equation}\label{theta_transform_sphere}
\left\{
\begin{aligned}
& \theta_1 = r\cos{\phi_1},\\
& \theta_2 = r\sin{\phi_1}\cos{\phi_2}, \\
& ~~~~~\vdots \\
& \theta_{m-1} = r\sin{\phi_1}\cdots, \sin{\phi_{m-2}}\cos{\phi_{m-1}}, \\
& \theta_{m} = r\sin{\phi_1}\cdots \sin{\phi_{m-2}}\sin{\phi_{m-1}}.
\end{aligned}
\right.
\end{equation}
The problem becomes calculating the PDF of $\theta$ on the region that 
\begin{equation}\label{approximation_P2_region}
\begin{aligned}
\Big\{[\phi_1, \cdots, &\phi_{m-1}, r]\in \mathbb{R}^m\Big|\ 0 \leq r\leq a, 0 \leq \phi_1 \leq \pi, \\ & 0 \leq \phi_2 \leq \pi, \cdots,
0 \leq \phi_{m-1} \leq 2\pi\Big\}.
\end{aligned}
\end{equation}

In the sphere coordinate, $\phi_1, \cdots, \phi_{m-1}, r$, of the region \eqref{approximation_P2_region} are divided into $n_1, n_2,\cdots,n_m$ units and the sizes are $\Delta{\phi_1}, \cdots, \Delta{\phi_{m-1}},\Delta{r}$, respectively. $\Delta{r}$ is set by 
\[\Delta{r} = \frac{\alpha}{n^{+}_m} > \frac{a}{n_m},\ n^{+}_m \in \mathbb{N}^{+}.\]
Use tuple $(k_1,\cdots,k_m),1\leq k_i\leq n_i$, to uniquely denote an unit. Then, it follows that
\begin{equation}
\left\{
\begin{aligned}
&\phi_{i, k_i} =\, \phi_{i,0} + k_i\Delta{\phi_i},\,\phi_{i,0} = 0,\,\phi_{i,n_i} = \pi,\\ & \qquad\qquad i = 1,2,\cdots, m-2, \\
&\phi_{m-1, k_{m-1}} =\,\phi_{m-1,0} + k_{m-1}\Delta{\phi_{m-1}},\,\phi_{m-1,0} = 0,\\& \qquad \qquad\phi_{m-1,n_{m-1}} = 2\pi, \\
& r_{k_m} =\, r_0 + k_m\Delta r,\,r_0 = 0, r_{n_m} = n_m\Delta{r}.
\end{aligned}
\right.
\end{equation}
The internal point $[\phi_1, \cdots, \phi_{m-1}, r]^{T}$ of one unit satisfies
\begin{equation}
\left\{
\begin{aligned}
\phi_{i, k_i-1} \leq\, &\phi_i \leq \phi_{i, k_i},\,i = 1,2,\cdots,m-1,\\
r_{k_m-1} \leq\, &r \leq r_{k_m}.
\end{aligned}
\right.
\end{equation}

Next, we give the discretization form of the objective function and constraints of $\bm{\mathrm {P^{2}_2}}$. Denote the probability that $\theta$ locates in $(k_1, k_2, \cdots, k_m)$-th unit as $p_{\{k_i\}_{i=1}^{m}}$. After discretization, the objective function becomes
\begin{equation*}\label{P3_objective}
\begin{aligned}
&J_c^{0}=
\sum_{k_1=1}^{n_1}\cdots\sum_{k_{m-1}=1}^{n_{m-1}}\sum_{k_{m}=1}^{n^{+}_{m}}
p_{\{k_i\}_{i=1}^{m}}.
\end{aligned}
\end{equation*}

As for the constraints, $p_{\{k_i\}_{i=1}^{m}}$ varies from $0$ to $1$ in each unit and the sum of all probabilities equals to one. This is described by
\begin{subequations}\label{P3_constraint1}
\begin{align}
0 \leq \,&p_{\{k_i\}_{i=1}^{m}}\leq 1, \\
\sum_{k_1, k_2, \cdots, k_m}&p_{\{k_i\}_{i=1}^{m}} = 1.
\end{align}
\end{subequations}

Denote the covariance of $\theta_i$ and $\theta_j$ generated by the unit $(k_1,\cdots,k_{m-2}, k_{m-1}, k_m)$ as $\Sigma_{ij,\{k_i\}_{i=1}^{m}}$. For the probability and covariance of each divided piece, it follows that
\begin{equation}\label{eq:pr_covar_inequality}
(\theta_i\theta_j)_{\min}p_{\{k_i\}_{i=1}^{m}}\leq\Sigma_{ij,\{k_i\}_{i=1}^{m}}\leq(\theta_i\theta_j)_{\max}p_{\{k_i\}_{i=1}^{m}}.
\end{equation}
We can use \eqref{theta_transform_sphere} to determine $(\theta_i\theta_j)_{\min}$ and $(\theta_i\theta_j)_{\max}$ easily in the corresponding unit. For example, if $i = 1, j = 2$ and $m>2$, it holds that
\begin{equation}
\begin{aligned}
(\theta_1\theta_2)_{\max} &= \left[\frac{1}{2}r^2\sin 2\phi_1\cos\phi_2\right]_{\max}\\
&= \frac{1}{2}r_{k_m}^2\left[\sin 2\phi_1\cos\phi_2\right]_{\max},
\end{aligned}
\end{equation}
where $\phi_1\in[\phi_{1,k_1-1}, \phi_{1,k_1}]$, $\phi_2\in[\phi_{2,k_2-1}, \phi_{2,k_2}]$ and $[\sin 2\phi_1\cos\phi_2]_{\max}$ can be obtained easily.

The first constraint in $\bm{\mathrm {P^{2}_2}}$ is that
$f_{\theta}(z)$ is symmetrical about the vector $y$. The discretization form is 
\begin{equation}\label{P3_constraint2}
p_{\,k_1,\cdots, k_{m-1}, k_m} = p_{\,n_1-k_1,\cdots, k_{m-1}+\frac{n_{m-1}}{2}, k_m}.
\end{equation}
The sum of covariances generated by all units equals to the covariances of whole space $\Omega$ and the second constraint of $\bm{\mathrm {P_2}}$ is rewritten as
\begin{equation}\label{P3_constraint3}
{\underline{\Sigma}}_{ij}\leq\Sigma_{ij} \rightarrow \sum_{k_1, k_2, \cdots, k_m}\Sigma_{ij,\{k_i\}_{i=1}^{m}} \leq {\overline{\Sigma}}_{ij}.
\end{equation}

The third constraint becomes 
\begin{equation}\label{P3_constraint_last}
p_{\,\{k_i\}_{i=1}^{m-1},\,k_m+1}\leq p_{\,\{k^{+}_i\}_{i=1}^{m-1},\,k_m}.
\end{equation}
After the discretization, we have 
\begin{equation}\label{eq:linear_opt}
\begin{aligned}
\bm{\mathrm {P^0_3}}: &\min\limits_{p_{\{k_i\}_{i=1}^{m}},\  \Sigma_{ij,\{k_i\}_{i=1}^{m}}} J_c^o\left(p_{\{k_i\}_{i=1}^{m}}, \Sigma_{ij,\{k_i\}_{i=1}^{m}}\right)\\
&~~~~~~~\text{s.t.}~~~~\eqref{P3_constraint1},\,\eqref{eq:pr_covar_inequality},\,\eqref{P3_constraint2},\, \eqref{P3_constraint3}\ \text{and}\, \eqref{P3_constraint_last}.\ 
\end{aligned}
\end{equation}

Note that $\bm{\mathrm {P^0_3}}$ is a linear optimization problem about decision variable $p_{\{k_i\}_{i=1}^{m}}$ and $\Sigma_{ij,\{k_i\}_{i=1}^{m}}$. Hence, its optimal solution can be obtained by the existing solver with polynomial time complexity. By solving $\bm{\mathrm {P_3^0}}$, we obtain an optimal distribution of $\theta(k)$ on $\Omega$.

\begin{remark}\label{rem2.1}
Comparing $\bm{\mathrm {P_1}}$ and $\bm{\mathrm {P_2^0}}$, the solution to $\bm{\mathrm {P_1}}$ only gives conditions on variance, while probability measure $\mathbb{P}(\left |\cdot \right|\leq \alpha^2)$ in $\bm{\mathrm {P_2^0}}$ provides the exact distribution.
\end{remark}

For $\bm{\mathrm{C}_2}$ in $\bm{\mathrm {P_2}}$, $f_{\theta}(z)$ is an optimal distribution iff $CA\varepsilon(k)+\theta(k)$ 
subject to the optimal distribution $f^{*}_{\theta}(z)
$. Since the PDF of $\varepsilon(k)$ is unknown, the optimal distribution of $\theta$ cannot be achieved. But the optimal distribution under $\bm{\mathrm{C}_1}$ is able to be taken to approximate that under $\bm{\mathrm{C}_2}$ for following reasons.
First, when $\mathbb{E}(\|CA\varepsilon(k)\|)\ll \mathbb{E}(\|\theta(k)\|)$, which is reasonable in practice, the solution to $\bm{\mathrm{C}_1}$ is approximately optimal for $\bm{\mathrm{C}_2}$. Second, $\varepsilon(k)$ will not degrade the performance of random input with arbitrary PDF $f_{\theta_m}(z)$ as shown in the following theorem.

\begin{theorem}\label{thm4}
{Let $\hat u^{*}_{1}(k)$ and $\hat u^{*}_{2}(k)$ be the optimal input predictions for $\hat x(k)\neq x(k)$ and $\hat x(k)= x(k)$, respectively. Then, there $\exists \,\alpha_1\in\mathbb{R}^{+}$, such that 
\begin{align*}
J_2(f_{\theta}(z),\hat u^{*}_{1}(k),\hat x(k),\alpha)\leq 
J_2(f_{\theta}(z),\hat u^{*}_{2}(k),\hat x^{*}(k),\alpha), 
\end{align*}
holds for $\forall\,\alpha \in (0,\alpha_1]$.}
\end{theorem}
\begin{proof}
The proof is given in Appendix \ref{app4}.
\end{proof}


\subsection{Calculations of Extra Inputs}
To introduce $\theta(k)$ into the state transitions of the system \eqref{eq1},
an extra control input $u_e(k)$ is added given by \eqref{extra_input}. This subsection gives the method to calculate extra input $u_e(k)$.

For $\bm{\mathrm {P_1}}$, the optimal distribution $f_{\theta}(z)$ satisfies \[\mathbb{D}(\theta_i)=\sigma_i^{2}, i = 1,2, \cdots, m,\] which holds true iff each component of $u_e(k)$ takes the maximum variance. Hence, with the variance, there is no need to generate $\theta(k)$ and calculate $u_e(k)$. $u_e(k)$ is directly generated with arbitrary distribution with maximum variances.

For $\bm{\mathrm {P_2}}$, the specific form of $u_{e}(k)$ is expected. We first obtain an optimal distribution of $\theta(k)$ by solving $\bm{\mathrm {P_3^0}}$. Then, $\theta(k)$ is generated according to the optimal distribution. $u_e(k)$ is calculated by the linear equations
\begin{equation}\label{linear_equation}
    B_{1}u_{e}(k) = \theta(k).
\end{equation}
Hence, when we generate $\theta(k)$, \eqref{linear_equation} should be solvable. The solution $u_{e}(k)$ is discussed as follows. 
\begin{itemize}
    \item If $\rank(B_1) = m$, the linear equations \eqref{linear_equation} are solvable, and the solution is
\begin{equation}\label{extra_input_rank1}
u_{e}(k) = (B^{T}_{1}B_1)^{\dagger}B^{T}_{1}\theta(k).
\end{equation}
\item  If $\rank(B_1) = b < m$, there exists 
\[T_1 \in \left \{T_1 \in \mathbb{R}^{m \times m}|\ T_1^{T}T_1 = I_{m \times m}\right\}\]  
such that 
\begin{equation}\label{ue_rank_b}
\begin{aligned}
T_{1}B_{1}u_{e} &= \left[
\begin{array}{c}
T_{11}\\
T_{21}
\end{array}\right]B_1u_{e} \\&= \left[
\begin{array}{c}
T_{11}B_1\\
0_{(m-b)\times q}
\end{array}\right]u_{e} = \left[
\begin{array}{c}
T_{11}\\
T_{21}
\end{array}\right]\theta,
\end{aligned}
\end{equation}
and $\rank(T_{11}B_1) = b$. 
To obtain $u_e(k)$, a serial of additional constraints $T_{21}\theta = 0_{m-b}$ should be considered in the solving of $f_{\theta}(z)$ by $\bm{\mathrm {P_3^0}}$. However, $\bm{\mathrm {P_3^0}}$ is a linear programming problem about $p_{\{k_i\}_{i=1}^{m}}$ and $\Sigma_{\{k_i\}_{i=1}^{m}}$, and the constraints are hard to be imposed directly. Hence, we take linear transformations on $\theta$ as $\theta^{+} = T_1\theta$ first to change the conditions into constraints on some components to be zero. Then, we will show that the optimal distribution of $\theta^{+}$ can be obtained from solving a similar but solvable optimization problem $\bm{\mathrm {P_3^0}}$.
\end{itemize}

Based on the above discussion, we let 
\begin{align*}
\theta^{+} = T_1\theta &= [\theta^{+}_1, \cdots, \theta^{+}_b, \cdots, \theta^{+}_m]^{T} \\&= [\theta_1^{+}, \cdots, \theta_b^{+}, 0, \cdots, 0]^{T}.
\end{align*}
Using the property of orthogonal matrices, one follows
\begin{align*}
\Omega_0&=\left\{\theta \in \mathbb{R}^m\Big|\|\theta\|^{2} \leq \alpha^{2}\right\} \\&= \left\{\theta^{+} \in \mathbb{R}^m\Big|\|\theta^{+}\|^{2} \leq \alpha^{2}\right\}.
\end{align*}
Clearly, the PDF of $\theta^{+}$ is symmetric and its covariance matrix satisfies
\begin{equation}\label{sigma_theta_plus}
\underline{\Sigma}^{+} = T_1\underline{\Sigma} T_1^{T}  \leq \Sigma^{+} \leq T_1\overline{\Sigma} T_1^{T} =  \overline{\Sigma}^{+}.
\end{equation}
The constraint of $\theta$ that $~\forall~\|z_1\| \leq \|z_2\|, f_{\theta}(z_1) \geq f_{\theta}(z_2)$ also holds for $\theta^{+}$ by following proposition.

\begin{proposition}\label{propo2}
For $\forall~z_1, z_2\in \mathbb{R}^{m}$, if $f_{\theta}(z_1)\geq f_{\theta}(z_2)$ holds for $\|z_1\|\leq \|z_2\|$, it holds that $f_{\theta^{+}}(z_1)\geq f_{\theta^{+}}(z_2)$, where  $\theta^{+} = T_1\theta$ and $T_1$ is an orthogonal matrix. 
\end{proposition}
\begin{proof}
Since $\|z_1\|\leq \|z_2\|$ and $T_1$ is an orthogonal matrix, one infers that 
\begin{equation}
\|y^{+}_1\| = \|T_1z_1\| = \|z_1\| \leq \|z_2\|= \|T_1z_2\| = \|z^{+}_2\|.
\end{equation}
Then PDF of $\theta^{+}$ follows that
\begin{equation}
f_{\theta^{+}}(z^{+}_1) = f_{\theta}(Tz_1) \geq f_{\theta}(Tz_2) = f_{\theta^{+}}(z^{+}_2).
\end{equation}
and we have completed the proof.
\end{proof}

To obtain the optimal extra input, similarly, we convert the coordinates $[\theta^{+}_1, \cdots, \theta^{+}_m]^T$ in Cartesian coordinate system into high dimensional sphere coordinate by
\begin{equation}\label{theta_plus_transform_sphere}
\theta^{+}_{i} = \left\{
\begin{aligned}
& r\cos{\phi_1^{+}},~i = 1,\\
&  r\sin{\phi_1^{+}}\cdots \sin{\phi_{i}^{+}}\cos{\phi_{i}^{+}},~i=2, \cdots, b-1,\\
& r\sin{\phi_1^{+}}\cdots \sin{\phi_{b-2}^{+}}\sin{\phi_{b-1}^{+}},~i=b,\\
& 0,~i = b+1, \cdots, m.
\end{aligned}
\right.
\end{equation}
Note the main differences between the coordinate of $\theta$  in \eqref{theta_transform_sphere} and that of $\theta^{+}$ in \eqref{theta_plus_transform_sphere} is that the last $b$ dimensions of $\theta^{+}$ equals zero.
Similarly, $\phi_1^{+}, \cdots, \phi_{b-1}^{+}, r$ are divided into $n_1, \cdots,n_b$ units and the sizes are $\Delta{\phi_1^{+}}, \cdots, \Delta{\phi_{b-1}^{+}},\Delta{r}$, respectively. Similar to the formulation of $\bm{\mathrm {P^{0}_3}}$. By replacing $m$ with $b$ and $\underline{\Sigma}$ with ${\underline{\Sigma}^{+}}$ ($\overline{\Sigma}$ with $\overline{\Sigma}^{+}$), the discrete optimization problem $\bm{\mathrm {P^{1}_3}}$ is set up to obtain an optimal distribution of $\theta^{+}$.
The objective function is 
\begin{equation*}\label{P3_1_objective}
\begin{aligned}
&J_c^{1}=
\sum_{k_1=1}^{n_1}\cdots\sum_{k_{b-1}=1}^{n_{b-1}}\sum_{k_{b}=1}^{n^{+}_{b}}p_{\{k_i\}_{i=1}^{b}}.
\end{aligned}
\end{equation*}
$p_{\{k_i\}_{i=1}^{b}}$ varies from $0$ to $1$ and the sum of all probabilities equals to one, i.e.,
\begin{subequations}\label{P3_1_constraint1}
\begin{align}
0 \leq \,&p_{\{k_i\}_{i=1}^{b}}\leq 1, \\
\sum_{k_1, k_2, \cdots, k_b}&p_{\{k_i\}_{i=1}^{b}} = 1.
\end{align}
\end{subequations}
For the probability and covariance of each divided piece, it follows that
\begin{equation}\label{eq:pr_covar_inequalitz_b}
(\theta_i\theta_j)_{\min}p_{\{k_i\}_{i=1}^{b}}\leq\Sigma_{ij,\{k_i\}_{i=1}^{b}}\leq(\theta_i\theta_j)_{\max}p_{\{k_i\}_{i=1}^{b}}.
\end{equation}
Since $f_{\theta^+}(z)$ is symmetrical about the vector $y^+$, we have
\begin{equation}\label{P3_1_constraint2}
p_{\,k_1,\cdots, k_{b-1}, k_b} = p_{\,n_1-k_1,\cdots, k_{b-1}+\frac{n_{b-1}}{2}, k_b}.
\end{equation}
Similar to  \eqref{P3_constraint3}, the covariances of $\theta^+$ satifies
\begin{equation}\label{P3_1_constraint3}
{\underline{\Sigma}}_{ij}^+\leq\Sigma_{ij}^+ = \sum_{k_1, k_2, \cdots, k_b}\Sigma_{ij,\{k_i\}_{i=1}^{b}} \leq {\overline{\Sigma}}_{ij}^+.
\end{equation}
By Proposition \ref{propo2}, let the PDF of $\theta^+$ decreases with the norm of $\theta^+$ and $p_{\{k_i\}_{i=1}^{b}}$ follows that
\begin{equation}\label{P3_1_constraint_last}
p_{\,\{k_i\}_{i=1}^{b-1},\,k_b+1}\leq p_{\,\{k^{+}_i\}_{i=1}^{b-1},\,k_b}.
\end{equation}
Then $\mathrm{P^1_3}$ is formulated as
\begin{equation}
\begin{aligned}
\bm{\mathrm {P^1_3}}: &\min\limits_{p_{\{k_i\}_{i=1}^{b}},\  \Sigma_{ij,\{k_i\}_{i=1}^{b}}} J_c^1\left(p_{\{k_i\}_{i=1}^{b}}, \Sigma_{ij,\{k_i\}_{i=1}^{b}}\right)\\
&~~~~~~~\text{s.t.}~\eqref{P3_1_constraint1},\,\eqref{eq:pr_covar_inequalitz_b},\,\eqref{P3_1_constraint2},\, \eqref{P3_1_constraint3}\ \text{and}\  \eqref{P3_1_constraint_last}.\ 
\end{aligned}
\end{equation}

\begin{remark}
{
$\mathrm{P^1_3}$ is a linear programming. The complexity of solving $\mathrm{P^1_3}$ is $\mathcal{O}((n+n_c)^{1.5}nL)$, where $n$ is the number of optimized variables, i.e., $n = \prod \limits_{i=1}^bn_i$. $n_c$ is the number of constraints and $L$ is the number of bits\cite{vaidya1989speeding}. This suffers the curse of dimensionality, but $\mathrm{P^1_3}$ can be solved offline. The PDF of random input is stored and used as look-up table for real-time control.}
\end{remark}

After solving $\bm{\mathrm {P^{1}_3}}$, we are able to achieve the distribution of $\theta^{+}$ and generate specific value of $\theta^{+}$ according to the distribution. Then, the extra input is obtained as follow. 

{\bf{Extra Input Calculation}}: From \eqref{ue_rank_b} and $\theta^+ =T_1\theta$, one obtains that
\begin{equation}
\begin{aligned}
T_{1}B_{1}u_{e} &= \left[
\begin{array}{c}
T_{11}\\
T_{21}
\end{array}\right]B_1u_{e}= \theta^+.
\end{aligned}
\end{equation}
Let $\theta^+_{1:b}$ be the subvector of $\theta$, which contains from the first element to the $b$-th element of $\theta^+$. Clearly, we have
\begin{equation}\label{ue_rankb_theta_plus}
T_{11}B_1u_e = \theta^+_{1:b}.
\end{equation}
Since $\rank(T_{11}B_1) = b$,  we have \eqref{ue_rankb_theta_plus} is solvable and the solution is
\begin{equation}\label{extra_input_rank2}
\begin{aligned}
u_e = ((T_{11}B_1)^{T}T_{11}B_1)^{\dagger}(T_{11}B_1)^{T}\theta^{+}_{1:b}.
\end{aligned}
\end{equation}

Finlay, the detailed way to generate the optimal unpredictable input is given by the Algorithm \ref{algo-1}. The inputs of the algorithm are desired variances of each component of $u_e$ and the rank of matrix $B_1$. The output of the algorithm is the value of extra input $u_e$. In the algorithm, there are two cases divided by the relationship between the rank of matrix $B$, i.e., parameter $b$, and the output dimension $m$. If $b < m$, the distribution $p_{\{k_i\}_{i=1}^{b}}$ of $\theta^{+}$ is obtained by $P^{1}_3$, and $\theta^{+}$ is obtained by its distribution. The extra input is generated by  \eqref{extra_input_rank2} with $\theta^{+}$. If $b = m$, the distribution $p_{\{k_i\}_{i=1}^{m}}$ of $\theta$ is obtained by $P^{0}_3$ and $\theta$ is obtained by its distribution. The extra input is generated by \eqref{extra_input_rank1} with $\theta$.

\begin{algorithm}[t]
    \caption{Unpredictable Control Generations}
    \label{algo-1}
    \begin{algorithmic}[1]
    \REQUIRE{$\sigma_u, b = \rank(B_1)$}
    \ENSURE{$u_e$}
    \IF {$b < m$}
    {\STATE Calculate $\underline{\Sigma}, \overline{\Sigma}$ by \eqref{sigma_theta} and $\underline{\Sigma}^{+}, \overline{\Sigma}^{+}$ by \eqref{sigma_theta_plus}
    \STATE Solve $\bm{\mathrm {P^{1}_3}}$ to achieve $p_{\{k_i\}_{i=1}^{b}}(f_{\theta^{+}}(z))$
    \STATE Use $p_{\{k_i\}_{i=1}^{b}}$ to generate $\theta^{+}$
    \STATE $u_e = ((T_{11}B_1)^{T}T_{11}B_1)^{\dagger}(T_{11}B_1)^{T}\theta^{+}_{1:b}$}
    \ELSE
    {
     \STATE Calculate $\underline{\Sigma}$ and $\overline{\Sigma}$ by \eqref{sigma_theta}
     \STATE Solve $\bm{\mathrm {P_3^0}}$ to achieve $p_{\{k_i\}_{i=1}^{m}}(f_{\theta}(z))$
     \STATE Use $p_{\{k_i\}_{i=1}^{m}}$ to generate $\theta$
     \STATE $u_e = (B_1^{T}B_1)^{\dagger}B_1^{T}\theta$
    }		
    \ENDIF
    \STATE \textbf{Return} $u_e$
    \end{algorithmic}
\end{algorithm}

{With the stochastic input added to the original input of a system, the outputs of the system are hard to be predicted and achieve the unpredictability under variance and probability measures. However, the extra input differs original evolution of the system and degrades the control performance with original controller. Thus, we will consider the control performance under unpredictable control in the following section.}

\section{Combined with Existing Control Framework}\label{seciv}
 In this section, we quantify the performance degradation by the extra input and illustrate how to combine unpredictable control method with existing control framework. Both traditional LQR control and cooperation control are considered. 
\subsection{Combined with Constrained LQR}
The classical constrained linear quadratic regulator (CLQR) problem is formulated as follows
\begin{equation}\label{eq:eq_lqr}
\begin{aligned}
&\min\limits_{x,u} \sum_{k = 0}^{\tau_1}\frac{1}{2}x(k)^{T}Q_kx(k)+q^{T}_kx(k)
+\frac{1}{2}u(k)^{T}R_ku(k)\\
&~\text{s.t.}\ x(k+1) = Ax(k) + Bu(k),\\
&~~\quad\ x(k) \in \mathcal{X}_k,~ u(k) \in \mathcal{U}_k.
\end{aligned}
\end{equation}
where $Q_k$ and $R_k$ are positive definite matrices and $\mathcal{X}_k,\mathcal{U}_k$ are convex sets. 
By solving this problem, we are able to get $u^{*}$ and take $u^{*}(0)$ as the control policy at current time. 

{To make the system outputs unpredictable, we solve the problem \eqref{eq:eq_lqr} first and add an extra input according to the input and state constraints. The constraints on unpredictable control are written as $u_{e}(k) \in \mathcal{U}_{e,k}$. We bound $u_{e}(k)$ according to Chebyshev inequality:
$$\mathbb{P}(|u_{e,i}(k)|\geq \lambda \sigma_{u,i})\leq \frac{1}{\lambda^2},~ i = 1, 2, \cdots, p.$$
We choose $\lambda = 10$ as an example to ensure a 99\% probability and let $\lambda \sigma_u \in \mathcal{U}_{e,k}$. This will guarantee the random input and system satisfying constraints. $\sigma_u$ is chosen by trade-off between unpredictability and extra cost in objective function.
The distribution of $x(k)$ is denoted by 
$x(k)\sim\mathcal{D}(\mu_k, \Sigma_{x,k})$. Since $u_{e}(k) \sim \mathcal{D}(0, \Sigma)$, the evolution of the mean and
covariance of $x(k)$ are
\begin{equation}
\left\{
\begin{aligned}
&\mu_{k+1}=\,A\mu_k+Bu(k),\\
&\Sigma_{x,k+1}=\,A\Sigma_{x,k}A^{T}+\Sigma.
\end{aligned}
\right.
\end{equation}
The extra cost by the stochastic input is first quantified. For the expectation of the objective function,
\begin{equation}\nonumber
\begin{aligned}
\sum_{k = 0}^{\tau_1}\mathbb{E}\left[\frac{1}{2}x(k)^{T}Q_kx(k)+q_kx(k)
+\frac{1}{2}{u^{+}(k)}^{T}R_ku^{+}(k)\right],
\end{aligned}
\end{equation}
where $u^{+}(k) = u(k)+u_e(k)$ and} the term satisfy
\begin{equation}\nonumber
\begin{aligned}
&\mathbb{E}\!\left[{u^{+}(k)}^{T}R_ku^{+}(k)\right]\!=\!u(k)^{T}R_ku(k)\!+\!\mathbb{E}\left[u_{e}^{T}(k) R_k u_{e}(k)\right]\\
&=\,u^{T}(k) R_k u(k)+\tr\left(\mathbb{E}\left[\sqrt{R_k}u_{e}(k)(\sqrt{R_k}u_{e}(k))^{T}\right]\right)\\
&=\,u^{T}(k) R_ku(k)+\tr\left(\mathbb{E}\left[\sqrt{R_k}u_{e}(k)u_{e}^{T}(k)\sqrt{R_k}\right]\right)\\
&=\,u^{T}(k) R_ku(k)+\tr\left(\sqrt{R_k}\Sigma\sqrt{R_k}\right)\\
&=\,u^{T}(k) R_ku(k)+\tr\left(\Sigma R_k\right).\\
\end{aligned}
\end{equation}
Similarly
\begin{equation}
\begin{aligned}
\mathbb{E}\left[{x}^{T}(k) Q_kx(k)\right]=\mu_k^{T}Q_k\mu_k+\tr\left(\Sigma_{x,k} Q_k\right).
\end{aligned}
\end{equation}
{Therefore, the extra cost is changed to
\begin{equation}\label{je}
\begin{aligned}
J_e = \frac{1}{2}\tr\left(\Sigma_{x,k} Q_k\right)+\frac{1}{2}\tr\left(\Sigma R_k\right),
\end{aligned}
\end{equation}
which is generated by the extra input.} 

Note that this cost is influenced by the covariance matrix of the extra input. Since the covariance matrix $\Sigma$ also determines the unpredictability according to Theorem \ref{thm1}, $\Sigma$ is set to be a decision variable to achieve a trade-off between control performance and privacy of the output data. 
Lastly, the optimization problem is set up as
\begin{equation}\label{covariance_opt1-1}
\begin{aligned}
&\min_{\sigma_u}w_1J_e-w_2\min_{\hat u(k)}\mathbb{E}\left(\|\varepsilon_{y}(k+1)\|_2^{2}\right),\\
&\text{s.t.} \quad 0\leq\sigma_{u,i}\leq\overline{\sigma}_{u,i},~i=1,2,\cdots,p.
\end{aligned}
\end{equation}

The solution to problem given by \eqref{covariance_opt1-1} is the optimized $\sigma_{u}$, which is one of the inputs of Algorithm \ref{algo-1}. Then, the unpredictable control is generated with considerations of control performance. To solve this problem, we
substitute the expression of $J_e$ from \eqref{je} to \eqref{covariance_opt1-1} and with Theorem \ref{thm1}, and the optimization problem becomes
\begin{equation}\label{covariance_opt1-2}
\begin{aligned}
&\min_{\sigma_u} -w_2\sum_{i=1}^{p}\sigma_i^2+
\frac{w_1}{2}\sum_{i=1}^{p}(r_{ii}+q_{ii})\sigma_{u,i}\\
&\text{s.t.} \quad 0\leq\sigma_{u,i}\leq\overline{\sigma}_{u,i},~i=1,2,\cdots,p,
\end{aligned}
\end{equation}
where $r_{ii}, q_{ii}$ are $i$-th row $i$-th column item of $R,Q$.
This is a standard quadratic optimization and is able to be solved by {the state-of-art solvers like Gurobi \cite{gurobi}}.

\subsection{Combined with Cooperative Control}
When unpredictable control is adopted by cooperative agents, since random states of one agent have an effect on others by interactions, the convergence of cooperative control is degraded inevitably. This part quantifies the performance degradation and illustrates how to choose variances for each agent to achieve a trade-off between unpredictability and cooperation performance.

{A graph is defined to represent the communication topology among cooperative agents, Let $\mathcal{G = (V, E)}$ be a directed graph with vertex set $\mathcal{V} = \{v_1, v_2, \cdots, v_N\}$ and edge set $\mathcal{E \subseteq V \times V}$.  $(v_i,v_j) \in \mathcal{E}$ indicates that $v_j$ can receive information from $v_i$. The neighbor set of agent $i$ is denoted by $\mathcal{N}_i$, where $v_j \in N_i$ iff $(v_i,v_j) \in \mathcal{E}$. The adjacency matrix is defined as $A^{+}=[a^{+}_{ij}]\in \mathbb{R}^{N \times N}$ with $a^+_{ij}=1$ for $(v_i,v_j) \in \mathcal{E}$  and $a^+_{ij}=0$ for otherwise. The Laplacian matrix is $ L = D-A^+ $, where $ D = \mathrm{diag} (d_1,\ldots,d_N)$ with $d_i = \sum\limits_{j = 1}^{N}a^+_{ij}$.} 

{Considering the unpredictability and the cooperation control, the dynamic of each agent $i$ is given by
\begin{equation}
\begin{aligned} 
x_i(k+1) = A^+x_i(k) + Bu_i(k) + \theta_i(k),
\end{aligned}
\end{equation}
where $\theta_i(k)$ is the unpredictable control input. Usually, $u_i(k)$ is a linear feedback designed by
\begin{equation}\label{eq:cooperative_control}
\begin{aligned}
u_i(k) = \gamma K\sum_{j\in \mathcal{N}_i}a^{+}_{ij}(x_j - x_i),
\end{aligned}
\end{equation}
where $K$ is the feedback matrix and $a^{+}_{ij}$ is the element of adjacency matrix.} Hence, the collective closed-loop dynamics for the system is
\begin{equation}\label{collective_eq}
\begin{aligned}
x_c(k+1) =&\,\left[ (I_N\otimes A^+) - \gamma L \otimes BK\right]x_c(k)+\theta_c(k)\\
 =&\,A_cx_c(k)+\theta_c(k),
\end{aligned}
\end{equation}
where $x_c(k) = [x_1^{T}(k), x_2^{T}(k),\cdots, x_N^{T}(k)]^{T}$ and $\theta_c(k) = [\theta_1^{T}(k), \cdots, \theta_N^{T}(k)]^{T}$. 
When $\theta_c(k) = 0$, there are some existing methods \cite{lewis2013cooperative} to select the feedback matrix $K$ to guarantee $A_c$ (marginally) stable and consensus.

{To quantify the performance degradation in the convergence of cooperative control, the cooperative index is introduced as
\begin{equation}
\begin{aligned}
J_{co}(x_c) = \frac{1}{2}x^{T}_cQx_c+q^{T}x_c.
\end{aligned}
\end{equation}
The system achieves consensus iff  $J_{co}(x_c)$ takes the minimum value $0$, i.e., $J_{co}^{*} = 0$. }

When the stochastic input $\theta_c(k)$ is added, the distribution of $x_c(k)$ is denoted by 
$x_c(k)\sim\mathcal{D}_{c}(\mu_c(k), \Sigma_c)$. Since $\theta_i \sim \mathcal{D}(0, \Sigma_i)$, under the dynamics (\ref{collective_eq}), the evolution of the mean and variance of $x_c(k)$ are 
\begin{equation}
\left\{
\begin{aligned}
&{\mu_c(k+1)=\,A_c\mu_c(k)},\\
&\Sigma_c(k+1)=\,A_c\Sigma_c(k)A^{T}_c+\Lambda,
\end{aligned}
\right.
\end{equation}
where $\Lambda=\mathrm{diag}(\Sigma_1,\cdots,\Sigma_N)\in \mathbb{R}^{nN\times nN}$. 

{Denote the noiseless cooperative performance as $J_{co}^{-}(x_c)$.
Considering the unpredictable design, the performance degradation by unpredictable control is
\begin{equation}
\begin{aligned}
\Delta J_{co} = \mathbb{E}\left[J_{co}(x_c)\right]-\mathbb{E}\left[J_{co}^{-}(x_c)\right] =\frac{1}{2}\tr(\Sigma_cQ).
\end{aligned}
\end{equation}
}
When $A_c$ is asymptotically stable, $\mu_c(k)$ and $\Sigma_c(k)$ are convergent with time step $k$, i.e., 
\begin{align*}
\lim \limits_{k \to + \infty}\mu_c(k) = \mu^{*}, \lim \limits_{k \to + \infty}\Sigma_c(k) = \Sigma_c^{*}.
\end{align*}
When $A_c$ is marginally stable, $\mu^{*}$ exists and is infinite, but $\lim \limits_{k \to + \infty}\Sigma_c(k)$
may be unbounded and not exists \cite{katewa2018privacy}. For the cases when both $\mu^{*}$ and $\Sigma_c^{*}$ exist, we have
\begin{equation}\label{cooperation_degrade}
\begin{aligned}
\Delta J_{co}^{*} = \lim \limits_{k \to + \infty}\Delta J_{co} = \frac{1}{2}\tr(\Sigma^{*}_cQ),
\end{aligned}
\end{equation}
where $\Sigma^{*}_c$ is the solution to following Lyapunov equation
\begin{equation}
\begin{aligned}
\Sigma_c = A_c\Sigma_cA_c^{T}+\Lambda.
\end{aligned}
\end{equation}

{To design the variances of input in Algorithm \ref{algo-1} for each agent, three factors, i.e., cooperation performance, extra energy consumption and unpredictability, are taken into consideration. 
The optimal variances $\sigma_i^j$ are obtained from solving the following optimization problem.}
\begin{equation}\label{covariance_opt2-2}
\begin{aligned}
&\min_{\sigma^1_u,\cdots,\sigma^N_u}\frac{w_1}{2}\tr\left(\Sigma_c^{*}Q\right)+
\frac{w_2}{2}\sum_{j=1}^{N}\tr\left(\Sigma_{u}^jR\right)
-w_3\sum_{j=1}^{N}\sum_{i=1}^{p}\sigma^{j}_{u,i}\\
&\quad \text{s.t.} \quad 0\leq\sigma^{j}_{i}\leq\overline{\sigma}^{j}_{i},~i = 1,\cdots,p,~j = 1,\cdots,N.
\end{aligned}
\end{equation}
Similar to problem \eqref{covariance_opt1-2}, problem \eqref{covariance_opt2-2} is a quadratic optimization problem about elements in the covariance matrices $\Sigma_u^j$, which can be solved by existing solvers.

\section{Simulation Results}\label{secv}
In this section, referring to real-word applications and model used in \cite{le2013differentially}, we fist take the second-order integral model to verify theoretical results. It is shown how to use algorithm \ref{algo-1} to calculate an optimal distribution for $\bm{\mathrm {P_2}}$. Then, we simulate a single system with unpredictable control input and demonstrate the outputs of the system as paths on 2-D plane intuitively. Effects of random input variances and distributions on $\varepsilon_y(k+1)$ are studied. Next, we combine unpredictable control with the constrained linear quadratic regulator to achieve the trade-off between unpredictability and control performance. Finally, we illustrate multi-agent cooperation with unpredictable control. 

\subsection{System with Unpredictable Control Input}
Denote the control period as $T_s$. The second-order integral model with unpredictable control is described by
\begin{equation}\label{eq:2nd_order_model}
\begin{aligned}
x(k+1) =& \left[ \begin{array}{cccc}
    1 & 0 & T_s & 0 \\
    0 & 1 & 0 & T_s \\
    0 & 0 & 1 & 0 \\
    0 & 0 & 0 & 1
\end{array}\right]x(k)+
\left[\begin{array}{cc}
 \frac{1}{2}T_s^2 & 0 \\
 0 & \frac{1}{2}T_s^2\\
 T_s & 0\\
 0   & T_s\\
\end{array}
\right](u+u_e)\\
y =& \left[\begin{array}{cccc}
   1 & 0 & 0 & 0 \\
   0 & 1 & 0 & 0
\end{array}
\right]x(k).
\end{aligned}
\end{equation}

\noindent The system described by \eqref{eq:2nd_order_model} is controllable and observable. Let the control period be $T_s = 1s$ and the upper bound of the covariance of unpredictable control $u_e$ be $\overline{\Sigma}_u = \mathrm{diag}(\frac{1}{2},\frac{1}{2})$. Then, the upper bound of the covariance of $\theta$ is
$\overline{\Sigma} = \mathrm{diag}(\frac{1}{8}T_s^4,\frac{1}{8}T_s^4)$, and the lower bound is $\underline{\Sigma} = 0_{2\times2}$. 
Theorem \ref{thm1} shows that the unpredictability with variance measure is maximized iff the covariance of $\theta$ satisfies 
\[\Sigma =\overline\Sigma=\mathrm{diag}(\frac{1}{8}T_s^4,\frac{1}{8}T_s^4)\]
or iff $\Sigma_u = \overline{\Sigma}_u$. Next, we illustrate how to use Algorithm \ref{algo-1} to calculate the optimal distribution for $\bm{\mathrm {P_2}}$. 

According to \eqref{eq:sphere_boundary}, we can set $a=5\sqrt{\sum_{i=1}^{m}\sigma^2_i}=2.5$.
In the sphere coordinate, we divide $\phi, r$ into $n_1 = 1$ unit and $n_2 = 26$ units, respectively. Then, a linear optimization problem as  \eqref{eq:linear_opt} is set up to achieve the optimal distribution of $\theta$. By solving the problem, we obtain the probability $p_i$ and covariance $\Sigma_{ij}$ of each piece in sphere coordinate. One discrete form of the optimal distribution is shown in Fig. \ref{fig:theta_opt_pdf}. The random vector $(z_{1,i},z_{2,i})$ in piece $i$ is calculated by
\begin{equation}
    z_{1,i} = \pm \sqrt{\frac{\Sigma_{11,i}}{p_i}},\,z_{2,i} = \pm \sqrt{\frac{\Sigma_{22,i}}{p_i}}.
\end{equation}
\begin{figure}[htb]
 	\centering
 	\includegraphics[width=7.8cm,height=5.4cm]{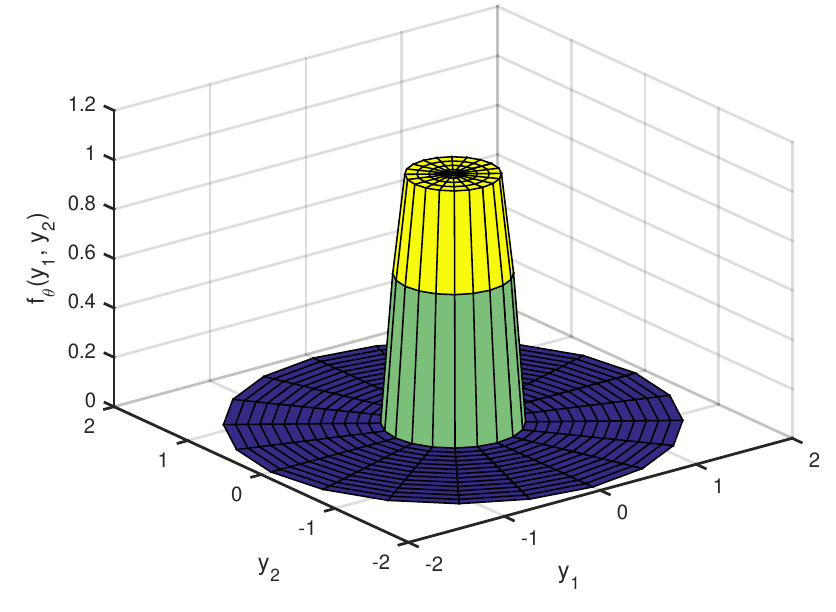}
 	\caption{An optimal distribution of $\theta$.}
    \vspace{-5pt}
 	\label{fig:theta_opt_pdf}
\end{figure}

We contrast the unpredictability of the calculated optimal distribution, an uniform distribution, with Gaussian and Laplace distribution, which are commonly used in the security of CPSs \cite{le2013differentially,he2018privacy}. The covariances of four types of distributions are all the same. The variance and probability metrics in $\bm{\mathrm{P_1}},\bm{\mathrm{P_2}}$ are taken and the results are shown in table \ref{table:pdf_contrast}. From the table, all distributions have the same variance metric. Among all distributions, the optimal distribution has the maximum unpredictability with probability metric. When $\alpha$ is small and takes $\alpha = 0.1,0.2,0.4$, the uniform distribution has the smallest probability metric and maximum unpredictability. This verifies Theorem \ref{thm2}. The uniform distribution has larger confidence probability than that of the optimal distribution. This is because the optimal distribution is not continuous, which is beyond the PDF compared with in Theorem \ref{thm2}.

\begin{table*}
\centering
\caption{\centering Unpredictability with input by four different distributions and without random input ($S = \|\varepsilon_{y}(k+1)\|_2^{2}\|$)}
\label{table:pdf_contrast}
\begin{tabular}{cccccc}
\toprule
Distribution & $\mathbb{E}(S)$ & $\mathbb{P}(S\leq0.1^2)$ & $\mathbb{P}(S\leq0.2^2)$ & $\mathbb{P}(S\leq0.4^2)$ & $\mathbb{P}(S\leq0.8^2)$ \\
\hline
Optimal PDF & 0.25 & \bf{0.0168} & \bf{0.0672}  & \bf{0.269} & \bf{0.779} \\
\hline
Uniform & 0.25  & 0.0209  & 0.0838 & 0.335 & 0.988 \\
\hline
Gaussian & 0.25 & 0.0392 & 0.1479 & 0.473 & 0.923\\
\hline
Laplace  & 0.25 & 0.0902 & 0.2632 & 0.590 & 0.903\\
\hline
No random & 0.0 & 1.0 & 1.0 & 1.0 & 1.0\\
\bottomrule
\end{tabular}
\end{table*}

\subsection{System with Unpredictable and LQR Control}
The system described by \eqref{eq:2nd_order_model} can denote a moving agent on 2-D plane. The inputs are accelerations along two axes and the states are positions and velocities. The initial state is $x(0) = [0, 0, 0, 0]^T$ and the expected end state is $x(t_f) = [20, 20, 0, 0]^T$. Suppose the system is originally governed by a LQR controller described by \eqref{eq:eq_lqr}. 
Suppose the measurement of attacker satisfies $\mathcal{N}(0,\mathrm{diag}(0.01, 0.01))$. Besides, the attacker uses Kalman filter algorithm to predict the position of the agent. In the algorithm, covariances of process noise and observation noise are set to be $\Sigma = \mathrm{diag}(1,1)$. The attacker achieves the optimal input prediction $\hat u^{*}(k)=u(k)$ at each time stamp.

Fig. \ref{fig_single_agent} illustrates complexity of agent motion with stochastic input compared to move without random input. Random input $\theta$ subjects to the optimal distribution with mean $0$ and upper covariance $\overline{\Sigma}_u = \mathrm{diag}(\frac{1}{2},\frac{1}{2})$. With unpredictable input, the trajectory of the agent becomes irregular and harder to be predicted accurately, even if the attacker has prior knowledge of the optimal distribution.

Fig. \ref{fig3} displays the relationship between unpredictability and covariances of the random input, which verifies the Theorem \ref{thm1}. Figures \ref{fig3.b} is achieved by smoothing data in Fig.\ref{fig3.a} by \[\overline S(k)=\frac{1}{3}(S(k-1)+S(k)+S(k+1)),\]
where $S = \|\varepsilon_{y}(k+1)\|_2^{2}\|$. 
\begin{figure}[t]
 	\centering
 	\includegraphics[width=6cm,height=5cm]{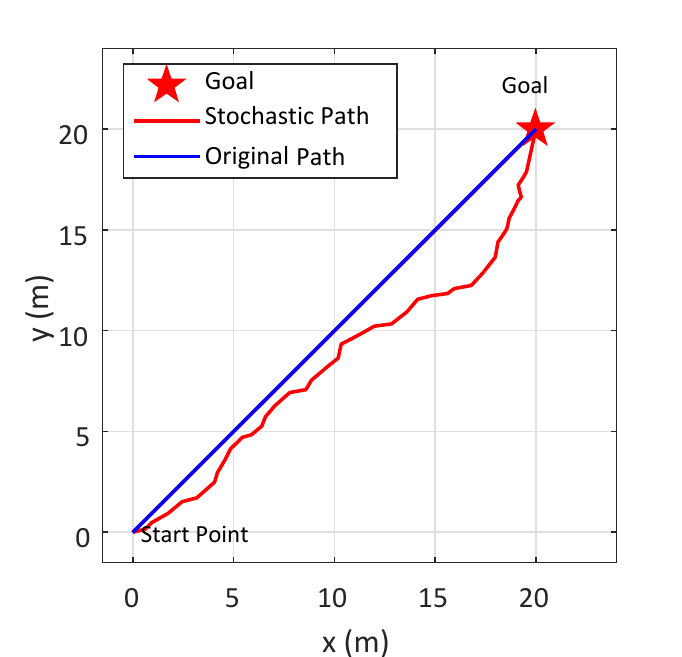}
 	\caption{Original and stochastic path of agent.}
    \vspace{-5pt}
 	\label{fig_single_agent}
\end{figure}
\begin{figure}[t]
\begin{center}
\subfigure[without smoothing.]{\label{fig3.a}
\includegraphics[width=0.23\textwidth]{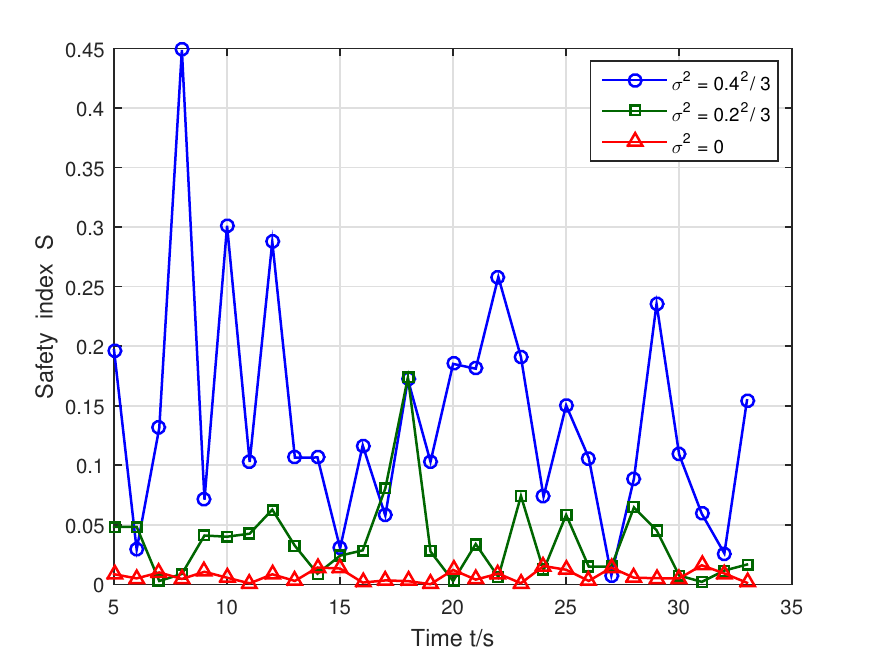}}
\subfigure[with smoothing.]{\label{fig3.b}
\includegraphics[width=0.23\textwidth]{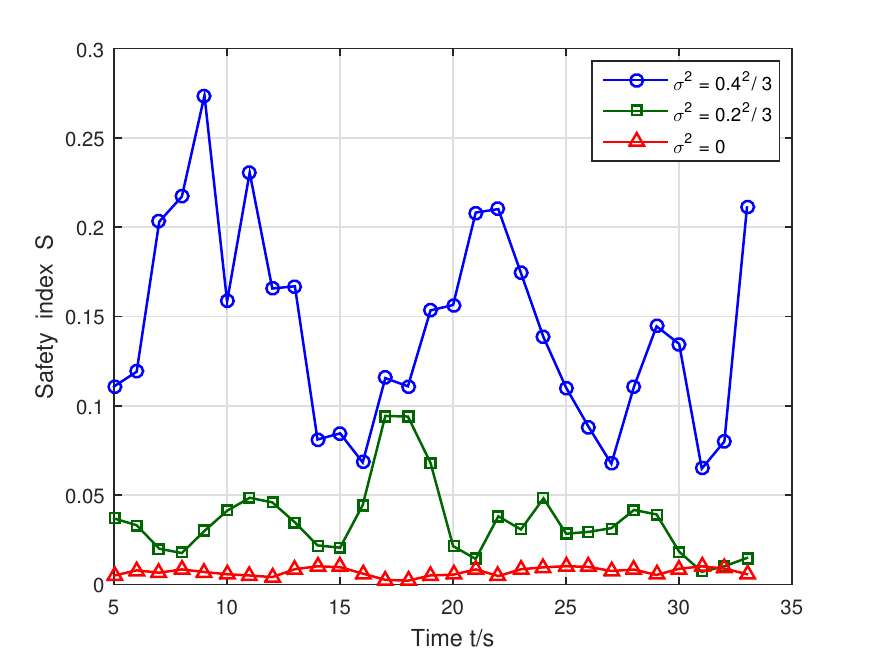}}
\caption{Prediction errors with the optimal distribution inputs with different covariances.}
\vspace{-10pt}
\label{fig3}
\end{center}
\end{figure}
\subsection{One Multi-agent System with Unpredictable Control and Cooperative Control Input}
Suppose the collective dynamics of the multi-agent system is described by  \eqref{collective_eq}. The cooperative control is described by  \eqref{eq:cooperative_control} with $\gamma_i=\frac{1}{2(1+d_i)}$ and $N=5$. The adjacency matrix $A^{+}$ and weight matrix $W$ are
\begin{gather*}
\centering
A^{+}=\left[
\begin{array}{ccccc}
0&0&0&0&1\\
1&0&1&0&1\\
0&1&0&1&0\\
0&0&0&0&1\\
0&0&0&0&0
\end{array}
\right]
\text{and} \ \,
W=\left[
\begin{array}{ccccc}
0&1&0&0&1\\
1&0&1&0&1\\
0&1&0&1&0\\
0&0&1&0&1\\
1&1&0&1&0
\end{array}
\right]_.
\end{gather*}
The initial positions of five agents are
(2,1), (-5,3), (-4,-3), (1,-3) and (0,0). The desired formation is described by
$\Delta^{1}=[-2,-4,-3,-1,0]^{T}$ and $\Delta^{2}=[2,0,-1.5,-1.5,0]^{T}$.
We set the same covariances $\Sigma$ for all agents.

Fig.~\ref{fig4} is the formation of five agents. Compared Fig.~\ref{fig4.a} with Fig.~\ref{fig4.b}, when random inputs are added, the formation is not formed and convergence level is degraded. Fig.~\ref{fig5.a} shows convergence level and performance degradation.
In Fig.~\ref{fig5.a}, the performance degradation $\Delta J_{co}$ is converged to $\Delta J_{co}^{*}$ with time. With unpredictable inputs, the error relative to desired formation $J_{co}$ fluctuates and performance degradation equals to deviation between expectation of $J_{co}$ and original $J_{co_0}$. Fig.~\ref{fig5.b} demonstrates$ \Delta J_{co}$ increase with the trace of covariance.
\begin{figure}[t]
\begin{center}
\subfigure[$t=30s$, formation control without unpredictable control.]{\label{fig4.a}
\includegraphics[width=4cm,height=3cm]{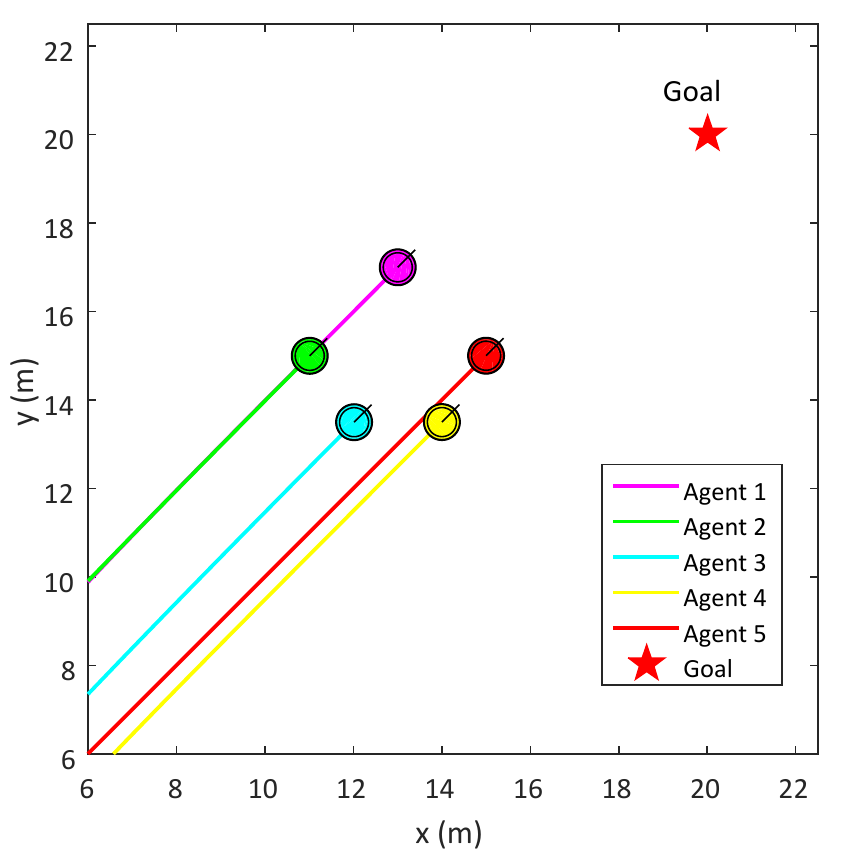}}
\subfigure[$t=30s$, formation control with unpredictable control.]{\label{fig4.b}
\includegraphics[width=4cm,height=3.1cm]{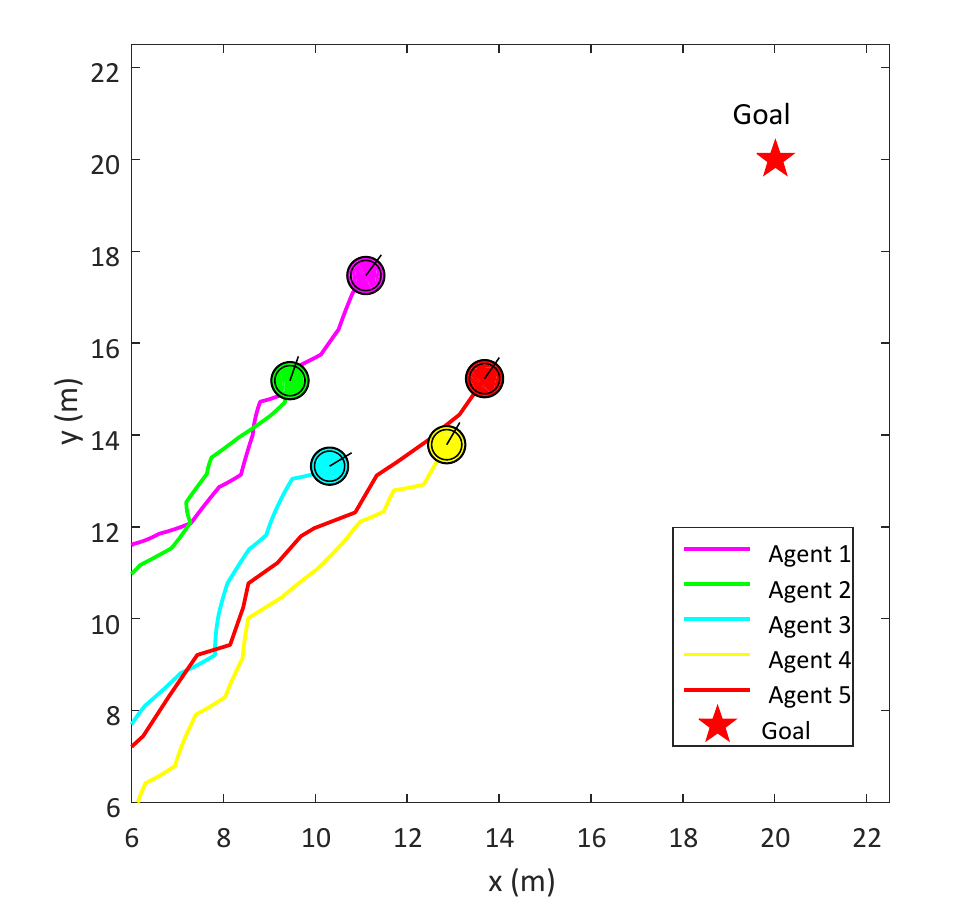}}
\caption{Illustration of $N=5$ agents in formation.}
\label{fig4}
\end{center}
\vspace*{-5pt}
\end{figure}

\begin{figure}[t]
\begin{center}
\subfigure[The convergence level and performance degradation.]{\label{fig5.a}
\includegraphics[width=0.23\textwidth]{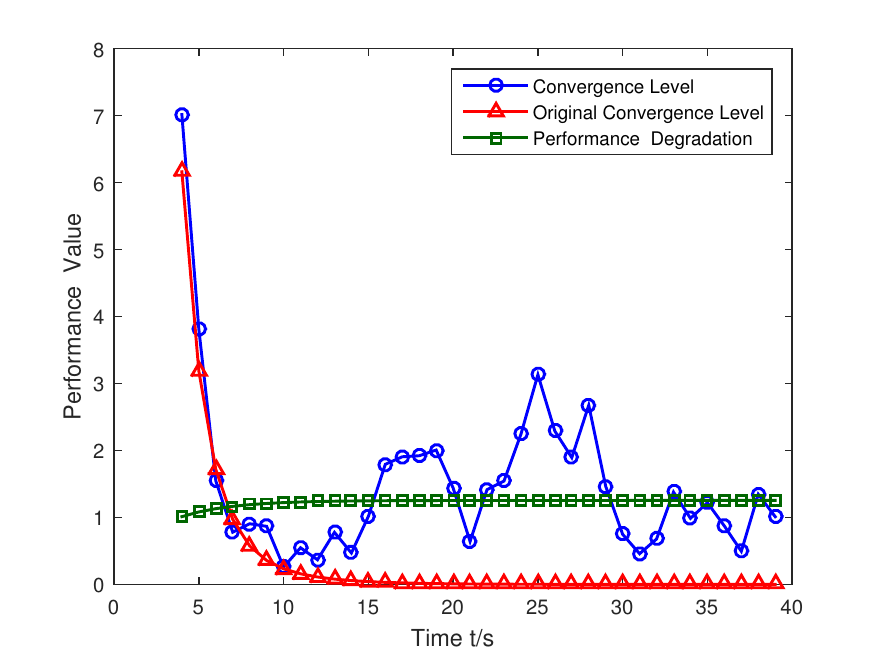}}
\subfigure[Performance degradations when the covariance takes different values.]
{\label{fig5.b}
\includegraphics[width=0.23\textwidth]{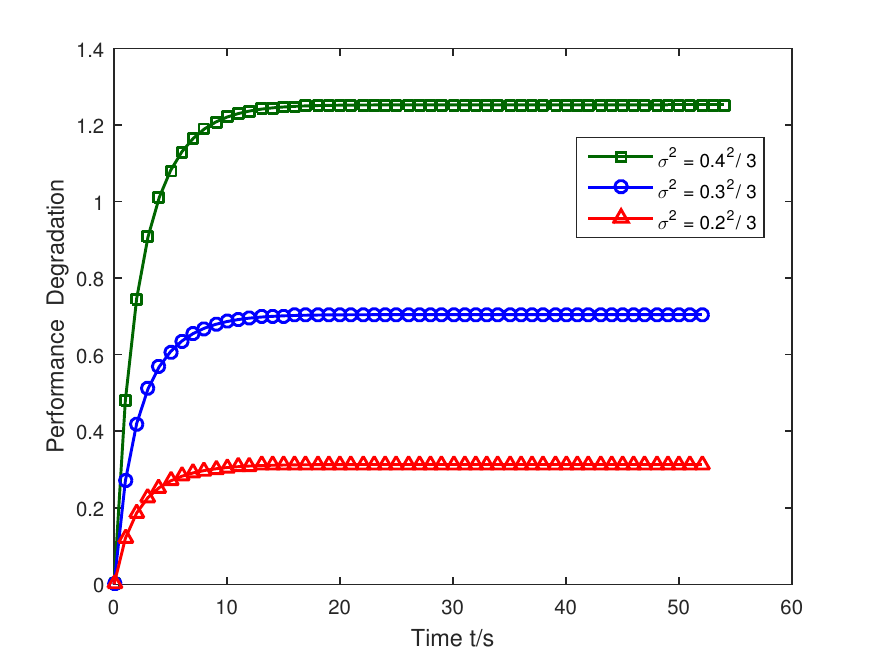}}
\caption{Performance degradation of formation convergence with random inputs.}
\label{fig5}
\end{center}
\vspace*{-5pt}
\end{figure}

\section{Conclusion}\label{secvi}
This paper investigates unpredictable control for the linear system. We quantify the unpredictability of the system with variance and probability metrics respectively. With variance metric, we prove that the unpredictability is proportional to the trace of covariance matrix. With probability metric, we prove that the uniform distribution of unpredictable control is the optimal among all continuous distribution when the confidence interval is small. We show how to calculate an optimal distribution with a novel numerical method. The solved optimal distribution outperforms Gaussian and Laplace distributions. We combine unpredictable control with LQR and cooperative control and demonstrate the performance by simulations. 
{Note that the current unpredictable design  only considers one step. In the future research we plan to extend it to the whole control horizon and achieve a trade-off between control performance and unpredictability over
the horizon. The proposed algorithm can be applied to mobile agents to preserve the security and safety from 
the external attacker who observes and predicts their trajectories.}


\appendix
\section{Proof of Theorem \ref{thm1}}\label{app0}
First, we consider $\bm{\mathrm{C}_1}$. The optimal distribution $f^{*}_{\theta}(z)$ is obtained by
solving $\bm{\mathrm {P_1}}$ under condition $\hat x(k)=x(k)$. Then, it follows that
\begin{align*}\label{eq-c1}
J_1= &{\mathbb{E}\left[ \left( {y(k+1)-\hat y(k+1)} \right)^{T}\left( {y(k+1)-\hat y(k+1)} \right) \right]} \nonumber \\
  = &\mathbb{E}\left[ \left\|C({Ax(k)+ Bu(k)+\theta_e(k)-A\hat x(k)-B\hat u(k)}) \right\|_{2}^{2} \right] \nonumber\\
  = &\mathbb{E}\left[ \left\|B_{1}\hat u(k) - B_{1}u(k)-\theta(k) \right\|_{2}^{2} \right] \nonumber
\end{align*}
{Denote $B_{1}\hat u(k) - B_{1}u(k)$ as random vector $X=[X_1,X_2,\cdots,X_n]^{T}$ and $-\theta(k)$ as $Z=[Z_1,Z_2,\cdots,Z_n]^{T}$, respectively. Each random variable $X_i$ is independent from random variable $Z_i$, and
$\mathbb{E}(Z_i)=0,i=1,2,\cdots,n$. It is not difficult to prove that $X + Z$ satisfies that
\[\mathbb{E}((X+Z)^{T}(X+Z))=\mathbb{E}(\left\|X\right\|_{2}^{2})+\mathbb{E}(\left\|Z\right\|_{2}^{2}).\]
Then, we have
\begin{align*}
J_1 
  = & \mathbb{E}(\left\|B_{1}\hat u(k) - B_{1}u(k)\right\|_{2}^{2}) + \mathbb{E}(\left\|\theta\right\|_{2}^{2})\\
  \geq & \mathbb{E}(\left\|\theta\right\|_{2}^{2})= \tr(\Sigma)=
  \sum_{i = 1}^{m}\mathbb{D}(\theta_i).
\end{align*}}

When $\hat u(k) = u(k)$, the objective function achieves the minimum value 
\begin{equation}\label{eq-d1}
\min_{\hat u(k)}{J_1}= \tr(\Sigma)=\sum_{i = 1}^{m}\mathbb{D}(\theta_i)\leq\sum_{i = 1}^{m}\sigma_i^2 = \tr(\overline\Sigma).
\end{equation}
According to the equality condition that  \eqref{eq-d1} holds, $f_{\theta}(z)$ is the optimal distribution if it makes $\tr(\Sigma_{ii})$ or $\mathbb{D}(\theta_i)$, $i = 1,2,\dots,m$, maximal in $\bm{\mathrm{C}_1}$.

Second, in $\bm{\mathrm{C}_2}$, $\varepsilon(k)$ is a random variable and the expectation of $\varepsilon(k)$ equals to zero. Since $\varepsilon(k)$ is the posterior estimation error by the attacker which is unknown, $\varepsilon(k)$ is independent with $\theta(k)$. Then, we have
\begin{equation}\label{eq12}
\begin{aligned}
\min_{\hat u(k)}J_1=\sum_{i = 1}^{m}\mathbb{D}(\theta_i)+\mathbb{E}\left[\|CA\varepsilon(k)\|_2^{2}\right].
\end{aligned}
\end{equation}

{{Since $\hat u(k)$ and $\mathbb{E}\left[\|CA\varepsilon(k)\|_2^{2}\right]$ are independent with each other, we have proved that $f_{\theta}(z)$ is the optimal distribution iff it makes each $\mathbb{D}(\theta_i)$ maximal.}}

\section{Proof of Theorem \ref{thm2}}\label{app1}
We first prove that when covariance matrix $\Sigma$ is fixed, under $\bm{\mathrm{C}_1}$, $\exists\ \alpha\in (0, \sqrt{3}\min \limits_{i}{\sigma_i}]$,  
the optimal distribution $f^*_{\theta}(z)$ is the multivariate uniform distribution. For uniform distributions, it is obvious that the larger input variance is, the smaller $J_2$ is. 

{By contradiction, we assume that $\forall~ \alpha < \sqrt{3}\min \limits_{i}{\sigma_i}$, there exists at least one of optimal $f^{*}_{\theta}(z)$ is not uniform distribution. According to definition \ref{def5}, we have
\begin{equation}\label{proof2-1}
\max_{\hat u(k)}\int_{\Omega}f^{U}_{\theta}(z)\,\mathrm{d} z \geq
\max_{\hat u(k)}\int_{\Omega}f^{*}_{\theta}(z)\,\mathrm{d} z,
\end{equation}
where $f^{U}_{\theta}(\cdot)$ represents a uniform distribution on $\Omega_U$.
{With $\alpha < \sqrt{3}\min \limits_{i}{\sigma_i}$}, $\hat u(k)$ makes $\int_{\Omega}f_{\theta}(z)\,\mathrm{d} z$ maximum such that $\Omega\subset \Omega_U$, where 
$\Omega_U=\{(z_1,\cdots,z_m)|
-\sqrt 3\sigma_{i} \leq z_i \leq \sqrt 3\sigma_{i}, i = 1, \cdots, m\}$.
}

Since Eq.\eqref{proof2-1} holds for arbitrary small $\alpha$, according to the continuity of $f_{\theta}(z)$, we have
\begin{equation}
\max_{z\in \Omega_U}f^{U}_{\theta}(z)\geq
\max_{z\in \Omega_U}f^{*}_{\theta}(z).
\end{equation}
{This shows that the maximum value in $\Omega_{U}$ of $f^{U}_{\theta}(z)$ is larger than $f^{*}_{\theta}(z)$. Consider the property of the uniform distribution. There is
\begin{equation}
f^{U}_{\theta}(z) \geq f^{*}_{\theta}(z),\,\forall z\in\Omega_U.
\end{equation}
}

Since variances of $f_{\theta}^{U}(z)$ and $f^{*}_{\theta}(z)$ are the same, it  means
\begin{equation}\label{thm3_1}
\begin{aligned}
&\int_{\Omega_U}f_{\theta}^{U}(z)z^{2}\,\mathrm{d} z =
\int_{\mathbb{R}^{m}}f^{*}_{\theta}(z)z^{2}\,\mathrm{d} z.
\end{aligned}
\end{equation}

When $z\in \mathbb{R}^{m}\setminus\Omega_U$, we have for $i=1,\cdots,m$,~$z_i^{2}> \sqrt{3}\sigma_i^{2}$.
{Then, it directly follows that
\begin{equation}\label{thm3_2}
\begin{aligned}
 &\int_{\Omega_U}(f_{\theta}^{U}(z)-f^{*}_{\theta}(z))\sum_{\i = 1}^{m}z_i^2\,\mathrm{d} z
=\int_{\mathbb{R}^{m}\setminus\Omega_U}f^{*}_{\theta}(z)\sum_{i = 1}^{m}z_i^2\,\mathrm{d} z\\
> & \sum_{i = 1}^{m}(\sqrt{3}\sigma_i)^{2}\int_{\mathbb{R}^{m}\setminus\Omega_U}
f^{*}_{\theta}(z)\,\mathrm{d} z.
\end{aligned}
\end{equation}}
Meanwhile, note that $~\sum_{i = 1}^{m}z_i^{2}\leq \sum_{i = 1}^{m}
(\sqrt{3}\sigma_i)^{2}, \forall~z\in \Omega_U$, thus it holds that
\begin{equation}\label{thm3_3}
\begin{aligned}
&\int_{\Omega_U}(f_{\theta}^{U}(z)-f^{*}_{\theta}(z))\sum_{i = 1}^{m}z_i^2
\,\mathrm{d} z\\
<& \sum_{i = 1}^{m}(\sqrt{3}\sigma_i)^{2}
\int_{\Omega_U}(f^{U}_{\theta}(z)-f^{*}_{\theta}(z)) \,\mathrm{d} z\\
=& \sum_{i = 1}^{m}(\sqrt{3}\sigma_i)^{2}
(1-\int_{\Omega_U}f^{*}_{\theta_m}(z)\,\mathrm{d} z)\\
=& \sum_{i = 1}^{m}(\sqrt{3}\sigma_i)^{2}
\int_{\mathbb{R}^{m}\setminus\Omega_U}f^{*}_{\theta}(z)\,\mathrm{d} z.
\end{aligned}
\end{equation}
Eqs. \eqref{thm3_2} and \eqref{thm3_3} renders a contradiction. 

Hence, one infers 
\[f^{*}_{\theta}(z)=f_{\theta}^{U}(z).\]

Therefore, one concludes that all the optimal distributions of $\theta$ should follow the uniform distribution.  We thus have completed the proof.

\section{Proof of Theorem \ref{thm3}}\label{app2}
We prove this theorem by the contradiction. 

Suppose $(f^{*}_{\theta}(z), 0)$ is not optimal and we represent one of optimal solutions as $(f^{*}_{\theta}(z), \delta)$, where $\delta \neq 0_m$. According to definitions in Section \ref{secii}, we have
\[\Omega_0=\left\{\theta \Big|\|\theta\| \leq \alpha\right\}\subset \Omega_1 = \left\{\theta \Big|\|\theta\| \leq a\right\}.\]
Let\[\Omega_{\delta} =\left\{\theta \Big| \|\theta - \delta\| \leq \alpha\right\}, ~\Omega_{-\delta} =\left\{\theta\Big|\|\theta + \delta\| \leq \alpha\right\} ,\]
and $\Omega_{\delta},~\Omega_{-\delta}\subset \Omega_1$.
Then, we have
\begin{equation}\label{eq:omega_delta_cmp}
\begin{aligned}
\int_{\Omega_0} f^{*}_{\theta}(z)\ \mathrm{d}{z} < \int_{\Omega_{\delta}} f^{*}_{\theta}(z)\ \mathrm{d}{z}.
\end{aligned}
\end{equation}

Define $f^{*,2}_{\theta}(z)$ as
\begin{equation}\label{proof_thm3.5_1}
f^{*,2}_{\theta}(z)=\left\{
\begin{aligned}
& \frac{1}{2}(f^{*}_{\theta}(y+\delta) + f^{*}_{\theta}(y-\delta)),~y \in \Omega_0, \\
& \frac{1}{2}(f^{*}_{\theta}(y-\delta) + f^{*}_{\theta}(z)),~y \in \Omega_{\delta}, \\
& \frac{1}{2}(f^{*}_{\theta}(y+\delta) + f^{*}_{\theta}(z)),~y \in \Omega_{-\delta}.
\end{aligned}
\right.
\end{equation}
One infers $\int_{\Omega_1}f^{*,2}_{\theta}(z)\ \mathrm{d}{z} = 1$, which means $f^{*,2}_{\theta}(z)$ is a PDF.

Next, we will prove that $f^{*,2}_{\theta}(z)$ satisfies the condition of the PDF of $\theta$. For condition i), we have $f^{*,2}_{\theta}(z) = f^{*,2}_{\theta}(-z)$. For condition ii), we denote the new variances of components as
$\mathbb{D}^{+}(\theta_i)$. It follows that for $i = 1, 2, \cdots, m$,
\begin{equation}
\begin{aligned}
\mathbb{D}^{+}(\theta_i)& - \mathbb{D}(\theta_i) = 
\int_{\Omega_0}\frac{1}{2}(f^{*}_{\theta}(y + \delta) + f^{*}_{\theta}(y - \delta))z_i^2~\mathrm{d}{z}\\
&+ \int_{\Omega_{\delta}}(f^{*}_{\theta}(z)+ f^{*}_{\theta}(y-\delta)) z_i^2~ \mathrm{d}{z} \\
& - \int_{\Omega_0}f^{*}_{\theta}(z)z_i^2~\mathrm{d}{z} - 2\int_{\Omega_{\delta}}f^{*}_{\theta}(z)z_i^2~\mathrm{d}{z}\\
&=  \int_{\Omega_0}f^{*}_{\theta}(y+\delta)z_i^2~\mathrm{d}{z}+\int_{\Omega_{\delta}}
f^{*}_{\theta}(y-\delta)z_i^2~\mathrm{d}{z} \\
& - \int_{\Omega_0} f^{*}_{\theta}(z)z_i^2~\mathrm{d}{z} - \int_{\Omega_\delta}f^{*}(z)_{\theta}z_i^2~\mathrm{d}{z}\\
& = \delta_i^2\left(\int_{\Omega_0}f^{*}_{\theta}(z)~\mathrm{d}{z}-\int_{\Omega_\delta}f^{*}_{\theta}(z)~\mathrm{d}{z}
\right).
\end{aligned}
\end{equation}
According to ~\eqref{eq:omega_delta_cmp}, we have 
\[\mathbb{D}^{+}(\theta_i)\leq \mathbb{D}(\theta_i).\] Thus, $f^{*,2}_{\theta}(z)$ satisfies condition ii). One concludes that $f^{*,2}_{\theta}(z)$ is the PDF of $\theta$. It follows that 
\[\int_{\Omega_0}f^{*,2}_{\theta}(z) \ \mathrm{d}{z} = 
\int_{\Omega_{\delta}}f^{*}_{\theta}(z) \ \mathrm{d}{z}.\]
If the following equation holds true, 
\[u^{*}(k) = \arg \max_{\hat u(k)}{\int_{\Omega}f^{*, 2}_{\theta}(z)} = 0,\]
then we have $(f^{*,2}_{\theta}(z), 0)$ is an optimal solution.

Suppose that $\delta_{2} = \arg \max_{\hat u(k)}{\int_{\Omega}f^{*, 2}_{\theta}(z)}\neq 0$. Then, it is clear that 
\[\int_{\Omega_{\delta_2}}f^{*,2}_{\theta}(z) \ \mathrm{d}{z} > 
\int_{\Omega_{\delta}}f^{*}_{\theta}(z) \ \mathrm{d}{z}.\]
The same conclusion with equation \eqref{proof_thm3.5_1}. We define $f^{*,3}_{\theta}(z)$ and 
we achieve 
\begin{equation}
\begin{aligned}
f^{*,n}_{\theta}(z) = {\lambda}^{n-1}\int_{\Omega_{\delta}}f^{*}_{\theta}(z) \ \mathrm{d}{z},~\lambda > 1,
\end{aligned}
\end{equation}
by iterations. When $n > 1 - \log_{\lambda}(\int_{\Omega_{\delta}}f^{*}_{\theta}(z)\mathrm{d}{z})$,  it holds that \[\int_{\Omega_{\delta_n}}f^{*}_{\theta}(z) \mathrm{d}{z} > 1,\]
which renders a contradiction.
Hence, $(f^{*,2}_{\theta}(z), 0)$ is an optimal solution, and we have completed the proof.

\section{Proof of Theorem \ref{thm4}}\label{app4}
Let the PDF of $\varepsilon_y(k)=CA(x(k)-\hat x(k))$ be $f_{\varepsilon}(z)$. According to the definition in Section \ref{secii}, we have
\[\Omega=\left\{\theta\Big|
\|\theta-\widetilde u\| \leq \alpha\right\}.\]Let set $\Omega_{\varepsilon}$ be
\[\Omega_{\varepsilon}\triangleq\left\{[\theta^T, \varepsilon^T]^{T}\in \mathbb{R}^{2m}\Big |\|
\theta_i-\widetilde{u}_i + \varepsilon_i\| \leq \alpha\right\}.\]

For arbitrary $\alpha >0$ and an optimal prediction $\hat u_{1}(k)$ for arbitrary $\hat x(k)\neq x(k)$, it follows that
\begin{align*}
&J_2(f_{\theta}(z),\hat u_{1}(k),\hat x(k),\alpha)\\
=&~\mathbb{P}\left(\| \varepsilon_y(k)+ \theta(k)+ B_{1}u(k) - B_{1}\hat u(k)
\|\leq \alpha\right)\\
=&~\int_{\mathbb{R}^{m}}f_{\varepsilon}(z)\left
(\int_{\Omega_{\varepsilon}}{f_{\theta}(z) \,\mathrm{d} z}\right)\,\mathrm{d}z\\
\leq&~\max_{\hat u(k)} \int_{\Omega}{f_{\theta}}(z)\,\mathrm{d} z \cdot \int_{\mathbb{R}^{m}}{f_{\varepsilon}(z)\,\mathrm{d}z}\\
=&~\max_{\hat u(k)} \int_{\Omega}{f_{\theta}}(z)\,\mathrm{d} z =J_2(f_{\theta}(z),\hat u^{*}_{2}(k),\hat x^{*}(k),\alpha).
\end{align*}
The equations above also holds for $\hat u_{1}^{*}(k)$. Thus, we have completed the proof.

\bibliographystyle{unsrt} 
\bibliography{ref}

\begin{thebibliography}{10}

\bibitem{li2020unpredictable}
Jialun Li, Jianping He, Yushan Li, and Xinping Guan.
\newblock Unpredictable trajectory design for mobile agents.
\newblock In {\em 2020 American Control Conference (ACC)}, pages 1471--1476. IEEE, 2020.

\bibitem{han2018privacy}
Shuo Han and George~J Pappas.
\newblock Privacy in control and dynamical systems.
\newblock {\em Annual Review of Control, Robotics, and Autonomous Systems}, 1:309--332, 2018.

\bibitem{le2013differentially}
Jerome Le~Ny and George~J Pappas.
\newblock Differentially private filtering.
\newblock {\em IEEE Transactions on Automatic Control}, 59(2):341--354, 2013.

\bibitem{mo2016privacy}
Yilin Mo and Richard~M Murray.
\newblock Privacy preserving average consensus.
\newblock {\em IEEE Transactions on Automatic Control}, 62(2):753--765, 2016.

\bibitem{he2018preserving}
Jianping He, Lin Cai, and Xinping Guan.
\newblock Preserving data-privacy with added noises: Optimal estimation and privacy analysis.
\newblock {\em IEEE Transactions on Information Theory}, 64(8):5677--5690, 2018.

\bibitem{manyam2019optimal}
Satyanarayana~Gupta Manyam, David Casbeer, Alexander Von~Moll, and Zachariah Fuchs.
\newblock Optimal dubins paths to intercept a moving target on a circle.
\newblock In {\em 2019 American Control Conference (ACC)}, pages 828--834. IEEE, 2019.

\bibitem{qu2022moving}
Chendi Qu, Jianping He, Jialun Li, Chongrong Fang, and Yilin Mo.
\newblock Moving target interception considering dynamic environment.
\newblock In {\em 2022 American Control Conference (ACC)}, pages 1194--1199. IEEE, 2022.

\bibitem{li2019learning}
Yushan Li, Jianping He, Cailian Chen, and Xinping Guan.
\newblock Learning-based intelligent attack against formation control with obstacle-avoidance.
\newblock In {\em 2019 American Control Conference (ACC)}, pages 2690--2695. IEEE, 2019.

\bibitem{roberts2004positional}
Stephen Roberts, Tim Guilford, Iead Rezek, and Dora Biro.
\newblock Positional entropy during pigeon homing i: {Application} of {Bayesian} latent state modelling.
\newblock {\em Journal of Theoretical Biology}, 227(1):39--50, 2004.

\bibitem{herbert2017escape}
James~E Herbert-Read, Ashley~JW Ward, David~JT Sumpter, and Richard~P Mann.
\newblock Escape path complexity and its context dependency in {Pacific} blue-eyes ({Pseudomugil} signifer).
\newblock {\em Journal of Experimental Biology}, 220(11):2076--2081, 2017.

\bibitem{isaacs1999differential}
Rufus Isaacs.
\newblock {\em Differential games: {A} mathematical theory with applications to warfare and pursuit, control and optimization}.
\newblock Courier Corporation, 1999.

\bibitem{mejia2019solutions}
Victor Gabriel~Lopez Mejia, Frank~L Lewis, Yan Wan, Edgar~N Sanchez, and Lingling Fan.
\newblock Solutions for multiagent pursuit-evasion games on communication graphs: Finite-time capture and asymptotic behaviors.
\newblock {\em IEEE Transactions on Automatic Control}, 2019.

\bibitem{teixeira2010cyber}
Andr{\'e} Teixeira, Saurabh Amin, Henrik Sandberg, Karl~H Johansson, and Shankar~S Sastry.
\newblock Cyber security analysis of state estimators in electric power systems.
\newblock In {\em 49th IEEE Conference on Decision and Control (CDC)}, pages 5991--5998. IEEE, 2010.

\bibitem{liu2011false}
Yao Liu, Peng Ning, and Michael~K Reiter.
\newblock False data injection attacks against state estimation in electric power grids.
\newblock {\em ACM Transactions on Information and System Security (TISSEC)}, 14(1):1--33, 2011.

\bibitem{pasqualetti2013attack}
Fabio Pasqualetti, Florian D{\"o}rfler, and Francesco Bullo.
\newblock Attack detection and identification in cyber-physical systems.
\newblock {\em IEEE Transactions on Automatic Control}, 58(11):2715--2729, 2013.

\bibitem{mo2009secure}
Yilin Mo and Bruno Sinopoli.
\newblock Secure control against replay attacks.
\newblock In {\em 2009 47th Annual Allerton Conference on Communication, Control, and Computing (Allerton)}, pages 911--918. IEEE, 2009.

\bibitem{irita2017detection}
Takashi Irita and Toru Namerikawa.
\newblock Detection of replay attack on smart grid with code signal and bargaining game.
\newblock In {\em 2017 American Control Conference (ACC)}, pages 2112--2117. IEEE, 2017.

\bibitem{teixeira2015secure}
Andr{\'e} Teixeira, Iman Shames, Henrik Sandberg, and Karl~Henrik Johansson.
\newblock A secure control framework for resource-limited adversaries.
\newblock {\em Automatica}, 51:135--148, 2015.

\bibitem{zhao2019resilient}
Chengcheng Zhao, Jianping He, and Qing-Guo Wang.
\newblock Resilient distributed optimization algorithm against adversarial attacks.
\newblock {\em IEEE Transactions on Automatic Control}, 65(10):4308--4315, 2019.

\bibitem{satchidanandan2016dynamic}
{Satchidanandan, Bharadwaj and Kumar, Panganamala R}.
\newblock {Dynamic watermarking: Active defense of networked cyber--physical systems}.
\newblock {\em {Proceedings of the IEEE}}, {105}({2}):{219--240}, {2016}.

\bibitem{geng2015optimal}
{Geng, Quan and Viswanath, Pramod}.
\newblock {The optimal noise-adding mechanism in differential privacy}.
\newblock {\em {IEEE Transactions on Information Theory}}, {62}({2}):{925--951}, {2015}.

\bibitem{duncan2000optimal}
{Duncan, George T and Mukherjee, Sumitra}.
\newblock {Optimal disclosure limitation strategy in statistical databases: Deterring tracker attacks through additive noise}.
\newblock {\em {Journal of the American Statistical Association}}, {95}({451}):{720--729}, {2000}.

\bibitem{nekouei2019information}
Ehsan Nekouei, Takashi Tanaka, Mikael Skoglund, and Karl~H Johansson.
\newblock Information-theoretic approaches to privacy in estimation and control.
\newblock {\em Annual Reviews in Control}, 47:412--422, 2019.

\bibitem{ariyo2014stock}
Adebiyi~A Ariyo, Adewumi~O Adewumi, and Charles~K Ayo.
\newblock Stock price prediction using the arima model.
\newblock In {\em 2014 UKSim-AMSS 16th International Conference on Computer Modelling and Simulation}, pages 106--112. IEEE, 2014.

\bibitem{delsole2004predictability}
Timothy DelSole.
\newblock Predictability and information theory. {Part} i: Measures of predictability.
\newblock {\em Journal of the Atmospheric Sciences}, 61(20):2425--2440, 2004.

\bibitem{delsole2005predictability}
Timothy DelSole.
\newblock Predictability and information theory. {Part} ii: {Imperfect} forecasts.
\newblock {\em Journal of the Atmospheric Sciences}, 62(9):3368--3381, 2005.

\bibitem{alahi2016social}
Alexandre Alahi, Kratarth Goel, Vignesh Ramanathan, Alexandre Robicquet, Li~Fei-Fei, and Silvio Savarese.
\newblock Social lstm: Human trajectory prediction in crowded spaces.
\newblock In {\em Proceedings of the IEEE conference on Computer Vision and Pattern Recognition (CVPR)}, pages 961--971, 2016.

\bibitem{shi2021sgcn}
Liushuai Shi, Le~Wang, Chengjiang Long, Sanping Zhou, Mo~Zhou, Zhenxing Niu, and Gang Hua.
\newblock Sgcn: Sparse graph convolution network for pedestrian trajectory prediction.
\newblock In {\em Proceedings of the IEEE/CVF Conference on Computer Vision and Pattern Recognition (CVPR)}, pages 8994--9003, 2021.

\bibitem{li2021spatio}
Jiachen Li, Hengbo Ma, Zhihao Zhang, Jinning Li, and Masayoshi Tomizuka.
\newblock Spatio-temporal graph dual-attention network for multi-agent prediction and tracking.
\newblock {\em IEEE Transactions on Intelligent Transportation Systems}, 2021.

\bibitem{hanzon2001state}
Bernard Hanzon and Raimund~J Ober.
\newblock A state-space calculus for rational probability density functions and applications to non-gaussian filtering.
\newblock {\em SIAM Journal on Control and Optimization}, 40(3):724--740, 2001.

\bibitem{liu2016online}
Chenghao Liu, Steven~CH Hoi, Peilin Zhao, and Jianling Sun.
\newblock Online {ARIMA} algorithms for time series prediction.
\newblock In {\em Thirtieth AAAI Conference on Artificial Intelligence}, 2016.

\bibitem{zhang1998time}
Jun Zhang and Kim-Fung Man.
\newblock Time series prediction using rnn in multi-dimension embedding phase space.
\newblock In {\em SMC'98 Conference Proceedings. 1998 IEEE International Conference on Systems, Man, and Cybernetics (Cat. No. 98CH36218)}, volume~2, pages 1868--1873. IEEE, 1998.

\bibitem{hilborn2000chaos}
Robert~C Hilborn et~al.
\newblock {\em Chaos and nonlinear dynamics: {An} introduction for scientists and engineers}.
\newblock Oxford University Press on Demand, 2000.

\bibitem{nozari2017differentially}
Erfan Nozari, Pavankumar Tallapragada, and Jorge Cort{\'e}s.
\newblock Differentially private average consensus: Obstructions, trade-offs, and optimal algorithm design.
\newblock {\em Automatica}, 81:221--231, 2017.

\bibitem{kawano2020design}
Yu~Kawano and Ming Cao.
\newblock Design of privacy-preserving dynamic controllers.
\newblock {\em IEEE Transactions on Automatic Control}, 65(9):3863--3878, 2020.

\bibitem{owhadi2013optimal}
Houman Owhadi, Clint Scovel, Timothy~John Sullivan, Mike McKerns, and Michael Ortiz.
\newblock Optimal uncertainty quantification.
\newblock {\em Siam Review}, 55(2):271--345, 2013.

\bibitem{han2015convex}
Shuo Han, Molei Tao, Ufuk Topcu, Houman Owhadi, and Richard~M Murray.
\newblock Convex optimal uncertainty quantification.
\newblock {\em SIAM Journal on Optimization}, 25(3):1368--1387, 2015.

\bibitem{vaidya1989speeding}
Pravin~M Vaidya.
\newblock Speeding-up linear programming using fast matrix multiplication.
\newblock In {\em 30th annual symposium on foundations of computer science}, pages 332--337. IEEE Computer Society, 1989.

\bibitem{gurobi}
{Gurobi Optimization, LLC}.
\newblock {Gurobi Optimizer Reference Manual}, 2023.

\bibitem{lewis2013cooperative}
{Lewis, Frank L and Zhang, Hongwei and Hengster-Movric, Kristian and Das, Abhijit}.
\newblock {\em {Cooperative control of multi-agent systems: {Optimal} and adaptive design approaches}}.
\newblock {Springer Science \& Business Media}, {2013}.

\bibitem{katewa2018privacy}
Vaibhav Katewa, Fabio Pasqualetti, and Vijay Gupta.
\newblock On privacy vs. cooperation in multi-agent systems.
\newblock {\em International Journal of Control}, 91(7):1693--1707, 2018.

\bibitem{he2018privacy}
Jianping He, Lin Cai, Chengcheng Zhao, Peng Cheng, and Xinping Guan.
\newblock Privacy-preserving average consensus: {Privacy} analysis and algorithm design.
\newblock {\em IEEE Transactions on Signal and Information Processing over Networks}, 5(1):127--138, 2018.

\end{thebibliography}

\textbf{Chendi Qu} received the B.E. degree in the Department of Automation from Tsinghua University, Beijing, China, in 2021. 
She is currently working toward the Ph.D. degree with the Department of Automation, Shanghai Jiao Tong University, Shanghai, China. 
She is a member of Intelligent Wireless Networks and Cooperative Control group. 
Her research interests include robotics, security of cyber-physical system, and distributed optimization and learning in multi-agent networks. 

\textbf{Jianping He} 
(SM'19) is currently an associate professor in the Department of Automation at Shanghai Jiao Tong University. He received the Ph.D. degree in control science and engineering from Zhejiang University, Hangzhou, China, in 2013, and had been a research fellow in the Department of Electrical and Computer Engineering at University of Victoria, Canada, from Dec. 2013 to Mar. 2017. His research interests mainly include the distributed learning, control and optimization, security and privacy in network systems.

Dr. He serves as an Associate Editor for IEEE Tran. Control of Network Systems, IEEE Open Journal of Vehicular Technology, and KSII Trans. Internet and Information Systems. He was also a Guest Editor of IEEE TAC, International Journal of Robust and Nonlinear Control, etc. He was the winner of Outstanding Thesis Award, Chinese Association of Automation, 2015. He received the best paper award from IEEE WCSP'17, the best conference paper award from IEEE PESGM'17, the finalist for the best student paper award from IEEE ICCA'17, and the finalist best conference paper award from IEEE VTC20-FALL.

\textbf{Jialun Li} 
(S'19) received the B.E. degree in the School of Astronautics from Harbin Institue of Technology, Harbin, China, in 2019. He is currently working toward the M.S. degree with the Department of Automation, Shanghai Jiaotong University, Shanghai, China. 
He is a member of Intelligent of Wireless Networking and Cooperative Control group. 
His research interests include estimation theory and robotics.

\textbf{Xiaoming Duan} 
obtained his B.E. degree in Automation from the Beijing Institute of Technology in 2013, his Master’s Degree in Control Science and Engineering from Zhejiang University in 2016, and his PhD degree in Mechanical Engineering from the University of California at Santa Barbara in 2020. He is currently an assistant professor in the Department of Automation, Shanghai Jiao Tong University.

\end{document}